\documentclass[aps,pra,twocolumn,superscriptaddress,floatfix,nofootinbib,showpacs,longbibliography]{revtex4-1}
\pdfoutput=1
\usepackage[utf8]{inputenc}  
\usepackage[T1]{fontenc}     
\usepackage[british]{babel}  
\usepackage[colorlinks=true, citecolor=blue, urlcolor=blue]{hyperref}  
\usepackage{graphicx} 
\usepackage[babel]{microtype}  
\usepackage{amsmath,amssymb,amsthm,bm,amsfonts,mathrsfs,bbm} 

\usepackage{xspace}  
\usepackage{pgf,tikz}
\usepackage{xcolor}
\usepackage{multirow}
\usepackage{array}
\usepackage{bigstrut}
\usepackage{braket}
\usepackage{color}
\usepackage{natbib}
\usepackage{multirow}
\usepackage{mathtools}
\usepackage[normalem]{ulem}
\usepackage{float}
\usepackage[caption = false]{subfig}
\usepackage{xcolor,colortbl}
\usepackage{color}
\usepackage{tikz}
\usetikzlibrary{positioning,arrows.meta,calc}
\usepackage{amssymb}
\newcommand{\Tr}{\operatorname{Tr}}

\newcommand{\be}{\begin{equation}}
\newcommand{\ee}{\end{equation}}
\newcommand{\ba}{\begin{eqnarray}}
\newcommand{\ea}{\end{eqnarray}}
\newcommand{\ketbra}[2]{|#1\rangle \langle #2|}

\newtheorem{theorem}{Theorem}
\newtheorem{corollary}{Corollary}
\newtheorem{definition}{Definition}
\newtheorem{proposition}{Proposition}

\newtheorem{lemma}{Lemma}
\newenvironment{theorem'}
 {\expandafter\def\expandafter\thetheorem\expandafter{\thetheorem'}\theorem}
 {\endtheorem}
\usepackage{mathtools}

\begin{document}

\title{
Lost and found charge in quantum batteries} 
\author{Debanjan Dey Sarkar}
\affiliation{Physics and Applied Mathematics Unit, Indian Statistical Institute, Kolkata, India}

\author{Mallika Mondal}
\affiliation{Physics and Applied Mathematics Unit, Indian Statistical Institute, Kolkata, India}

\author{Preeti Parashar}
\affiliation{Physics and Applied Mathematics Unit, Indian Statistical Institute, Kolkata, India}

\author{Tamal Guha}
\affiliation{Mathematical Institute, Slovak Academy of Sciences,  Štefánikova 49, 814 73 Bratislava, Slovakia}
\affiliation{Cryptology and Security Research Unit, Indian Statistical Institute, 203 B.T. Road, Kolkata 700108, India}

\begin{abstract}
Quantum batteries are prone to losing their stored charge, when interacting with a thermal environment. However, getting a limited assistance from the thermal environment, is it possible to recover the charge back, in a reusable form?  Here, we answer this question, by involving an assistance of a suitable measurement performed on the environment. This framework resembles the structure of quantum instruments in the thermodynamic scenario. Our proposed framework involves two different kind of assistance from thermal environment - one by accessing only the thermal particle, actively participating in the interaction,  providing a \textit{weak} retrieval of charge, while the other involves local assistance from an additional quantum system purifying the thermal environment, resulting in the \textit{strong} retrieval of the lost charges. By setting the upper-bound on the amount of charges retrieved in each of these two situations, we report that their difference characterizes the amount of entanglement generated between the quantum battery and the reference system due to the thermal interaction. Finally, we exemplify the extreme instances of the difference between the weak and the strongly retrieved charges for qubit batteries.
\end{abstract}
\maketitle
\textit{Introduction.}-- The theoretical framework for extending the notion of thermodynamics in the microscopic world \cite{lenard1978thermodynamical,popescu2006entanglement,esposito2009nonequilibrium,brandao2013resource,horodecki2013fundamental,goold2016role} finds its applicative breakthrough in context of miniaturizing thermal devices \cite{linden2010small,horowitz2014quantum,klatzow2019experimental,mitchison2019quantum, barra2022quantum}. Quantum battery (QB), admittedly, the key ingredient to engineer those devices, can be considered as a tiny energy storage device \cite{alicki2013entanglement,binder2015quantum,campaioli2024colloquium}. The presence of  non-classical elements in QB further fueled ultra fast charging \cite{campaioli2017enhancing, ferraro2018high, barra2019dissipative, crescente2020ultrafast, ueki2022quantum, pokhrel2025large, shastri2025dephasing}, thereby reducing the energy dissipation \cite{adolina2019extractable, shi2022entanglement, lu2025topological}. However, the fruitful implementation of QBs in real-world thermodynamic devices requires to overcome the barriers, faced in various contexts \cite{friis2018precision, julia2020bounds, peng2021lower, tirone2024quantum, bai2024change,liu2506open,malavazi2025two}; One of the primary among which is the perturbations caused due to the thermal environment. In practice, when a QB interacts with a thermal environment, often referred to as the \textit{bath}, all the \(\alpha\)-free energies, including the conventional one (with \(\alpha\to\infty\)), stored in it dies out in time \cite{brandao2013resource}. Notably, the amount of free energy stored in a quantum battery, analogous to classical thermodynamics, often encountered as the amount of extractable work, on an average \cite{skrzypczyk2014work, morris2019assisted, guha2020thermodynamic} and hence can be quantified as the charge stored in it. Therefore, the interaction with the thermal environment renders an unavoidable charge leakage. The amount of charge contained in a QB, however, can be characterized in various ways depending upon the practical scenario of implementation. For instance, if the battery is guaranteed to be kept isolated from the environment, that is, evolves only under the energy preserving unitary, then the ergotropy \cite{Allahverdyan, castellano2024extended, halder2024operational} sufficiently quantifies the amount of charge stored in it. 

Nevertheless, by identifying the conventional free energy as a suitable charge quantifier, here we investigate the maximal amount of charge that could be retrieved in the discharged battery with a limited access to the thermal environment. The framework operationally mimics an instance of transporting a charged battery at a distant location, with possible interruption due to the thermal interaction, similar to those in \cite{guha2020thermodynamic, simonov2022work, simonov2025activation}. Within this setting, we investigate the maximal amount of charge that could be retrieved back in the quantum battery, by accessing the thermal environment, taking part in the interaction, in a limited way.

At this point, we recall that the assistance of environment in correcting quantum channels was first encountered in the work of Gregoratti and Werner \cite{gregoratti2003quantum} and subsequently found its relevance in the context of classical, as well as quantum information transmission \cite{hayden2005correcting, karumanchi2016classical, karumanchi2016quantum,pirandola2021environment, oskouei2021capacities, harraz2022quantum, harraz2023high, chowdhury2025minimal}. The limited access to environment further bears implicit implications with regard to correlated quantum channels \cite{macchiavello2002entanglement, bowen2004quantum,ball2004exploiting,banaszek2004experimental}, their interfering configuration \cite{chiribella2019quantum,guha2023quantum,lai2024quick, saha2025interference,rubino2021experimental} and also in error correction \cite{gregoratti2004quantum}. While the correction, that is perfect inversion of quantum channels, sufficiently renders tracing back the leaked charge to the QB, it is not necessary at all. More precisely, since multiple quantum systems may possess equal free energy with respect to a given thermal bath, it is indeed possible to retrieve the dissipated energy back from the environment, however without getting the exact quantum state back. 

In the context of thermodynamic charge reversal, we first model the thermal operations in terms of their isometric extensions. In the unitary picture, this involves a bipartite pure entangled state with local thermal marginals as the environment. One of these two subsystems actively participates in the thermal interaction while the other can be thought of as a purifying reference. The assistance from the thermal environment, therefore, can be characterized in two different ways: the \textit{weak} one, which involves only the interacting thermal particle and the \textit{strong} one involving both the particles. While assistance from environment, in general, may involve various causal constraints on accessibility \cite{winter2005environment,chowdhury2025minimal}, here our model involves only local measurements on the bath and the reference systems. Moreover, this minimal access to the thermal environment causes no change on the role of the thermal bath marginally. Under such a scenario we show that the strong assistance helps to find the lost charge in the QB optimally. On the other hand, with a weak assistance the charge reversal is always sub-optimal. In fact, the gap between the charge reversals under the strong and weak assistance quantifies the amount of entanglement generated by the thermal operation between the input quantum state and the reference state. 

\textit{Isometric extensions for thermal operations}--The resource theory of thermodynamics in presence of a thermal bath undertakes a singleton set of free states with  the temperature identical to that of the bath \cite{brandao2013resource}. 

To put it more formally, consider a quantum system governed by the system Hamiltonian \(H_s=\sum_i \epsilon_i \ketbra{\epsilon_i}{\epsilon_i}\), together with a bath of given Hamiltonian \(H_b=\sum_kE_k\ketbra{E_k}{E_k}\). Any thermal state corresponding to the bath Hamiltonian can be written as 
\begin{equation}\label{e1}
    \tau_{\beta}^b=\frac{e^{-\beta H_b}}{\Tr(e^{-\beta H_b})}=\frac{\sum_ke^{-\beta E_k}\ketbra{E_k}{E_k}}{\sum_ke^{-\beta E_k}}
\end{equation}
where, \(\beta=\frac 1{k_B T}\) is the inverse temperature; \(k_B\) being the Boltzmann constant and \(T\) being the temperature of the thermal bath. Accordingly, the free state of the system is identified as 
\[\tau_{\beta}^s=\frac{e^{-\beta H_s}}{\Tr(e^{-\beta H_s})}=\frac{\sum_ie^{-\beta \epsilon_i}\ketbra{\epsilon_i}{\epsilon_i}}{\sum_ie^{-\beta \epsilon_i}}.\]
In the rest of the paper, we will often refer to such states as the thermal state of inverse temperature \(\beta\).
Accordingly, any \(\beta\)-thermal operation $(\Lambda_{\beta} )$ is defined as :
\begin{equation}\label{e2}
    \Lambda_{\beta}(\rho_s)=\Tr_{b}[U_{sb}(\rho_s\otimes\tau_{\beta}^b)U_{sb}^{\dagger}],
\end{equation}
where, \(U_{sb}\) is the system-bath joint unitary, preserving the total energy, that is, 
\[[U_{sb},H_s\otimes I_b+I_s\otimes H_b]=0,\]
with \(I_x\) as the identity operator on the Hilbert space \(\mathcal{H}_x\). The input and output Hilbert spaces (\(\mathcal{H}_s\)) are taken to be same here; However, in general, they can be different, with suitable choice of the traced out Hilbert space  \cite{gour2024resources}.

Under such a thermal operation, all the \(\alpha\)-free energies, including the conventional one, with respect to the thermal bath of inverse temperature \(\beta\) are non-increasing in nature. That is,
\[F_{\beta}^{\alpha}(\rho_s)\geq F_{\beta}^{\alpha}(\Lambda_{\beta}(\rho_s)),\quad\forall\rho_s\in\mathcal{D}(\mathcal{H}_s),\] \\ where,
$\mathcal{D}(\mathcal{H}_s)$ is the set of all density operators on $\mathcal{H}_s $. The \(\alpha\)-free energy \(F_{\beta}^{\alpha}(\rho_x)=\Tr(H_x\rho_x)-\frac 1{\beta}S_{\alpha}(\rho_x)\), where \(S_{\alpha}(\rho_x)=\frac1{1-\alpha}\ln(\Tr(\rho_x^{\alpha}))\) is the \(\alpha\)-Renyi entropy of a quantum state \(\rho_x\), multiplied by \(\ln 2\). From now on,  we will replace the natural logarithm \(\ln\) by \(\log\) in the rest of the manuscript.

While the  analysis given below  holds for  most general thermal operations, for the sake of simplicity, we will drop the \textit{primed}-symbol over the output Hilbert spaces, assuming \(\mathcal{H}_{s'}\simeq\mathcal{H}_{s}\) and \(\mathcal{H}_{b'}\simeq\mathcal{H}_{b}\). Let us now consider the following bipartite purification of the thermal state of the bath, given in Eq. (\ref{e1}),
\begin{align}\label{e3}
    \ket{\phi^+_{\beta}}_{bR}:=\frac 1{\sqrt{Z}}\sum_k \exp(-\frac{\beta E_k}{2})\ket{E_k}_b\otimes\ket{E_k}_R
\end{align}
where, \(Z=\sum_k\exp(-\beta E_k)\) is the partition function.  Note that, while the basis choice for the reference (\(R\)) can be made arbitrary, considering the same Hamiltonian here we use \(\{\ket{E_k}_R\}_k\) to preserve local thermality. Moreover, among the two quantum systems \(b\) and \(R\), only the first one---being  a part of the thermal bath---actively participates in the thermal interaction. Accordingly, for any thermal operation \(\Lambda_{\beta}\), one can construct a class of suitable isometric extension (\(\mathcal{V}_{\Lambda_{\beta}}:\mathcal{H}_s\mapsto\mathcal{H}_{sbR}\)) such that, for an arbitrary quantum battery \(\ket{\psi}_s=\sum_
ka_k\ket{\epsilon_k}_s,\) with \(a_
k\in\mathbb{C}\text{ and }\sum_
k|a_k|^2=1\),
\begin{align}\label{e4}
    \mathcal{V}_{\Lambda_{\beta}}\ket{\psi}_s=(U_{sb}\otimes\mathbb{I}_R)(\ket{\psi}_s\otimes\ket{\phi^+_{\beta}}_{bR})=\sum_ka_k\ket{\phi_k^{\beta}}_{sbR}.
\end{align}
While the isometry and hence the global energy-preserving unitary \(U_{sb}\) in general, depends upon the complete structure of the system and the bath Hamiltonians, our main results do not involve them explicitly. Note that, for general quantum channels, there is complete freedom to choose the ancillary system and the corresponding global unitary. However, the action of thermal operations by definition (as in Eq. (\ref{e2})) restricts them to be a specific thermal state and an energy preserving unitary respectively. From this we conclude  that the form of Eq. (\ref{e4}) is unique up-to the local unitary operation on the reference system \(R\).

\textit{Retrieving the lost charge from thermal environment}-- We will now move to our main results, concerning the retrieval of the free energy dissipated to the thermal environment. 
Consider a d-level quantum battery \(\rho_s=\sum_{i=0}^{d-1}p_i\ketbra{\psi_i}{\psi_i}_s\) in its spectral form, where
\[\ket{\psi_i}=\sum_{k=0}^{d-1}a_k^i\ket{\epsilon_k},\text{ such that }\forall k,~a_k\in\mathbb{C}\text{ and }\sum_k|a_k|^2=1.\] 
When this quantum battery interacts with a thermal bath of inverse temperature \(\beta\), it gets correlated with both the bath and the reference quantum system. Using Eq. (\ref{e4}) we can obtain the following form of the system-bath-reference joint state
\begin{align}\label{e5}
    \sigma_{sbR}&= \sum_{i=0}^{d-1}p_i\ketbra{\phi_{\psi_i}^{\beta}}{\phi_{\psi_i}^{\beta}}_{sbR}\\\nonumber\text{with, }\ket{\phi_{\psi_i}^{\beta}}_{sbR}&=\sum_ka^i_k\ket{\phi_{k}^{\beta}}_{sbR},\text{ as in Eq. (\ref{e4})}.
\end{align}

The resource theory of thermodynamics then immediately implies
\[F_{\beta}^{\alpha}(\Tr_{bR}[\sigma_{sbR}])\leq F_{\beta}^{\alpha}(\rho_s),\]
rendering a loss in the stored charge in the quantum battery.
Our objective is now to retrieve this lost free-energetic charge in the quantum battery, by getting assistance from the outcome of the measurement performed on the bath and the reference qubits respectively (see Fig. \ref{f1}). At this point, it is important to mention that depending upon the framework of assistance and accordingly the work extraction protocols, one can define the retrieved charges in various ways. For instance, if the discharged battery is allowed to use in single shot fluctuation-free work extraction, then \(F_{\beta}^{\alpha=0}(\rho)\) is the relevant quantity of interest. However, in the present case we consider the amount of extractable work on an average from many copies of the quantum state, which is characterized by the traditional free energy \(F_{\beta}^{\alpha\to 1}(\rho)\) \cite{skrzypczyk2014work}. From now on, we will represent our specific quantity of interest \(F_{\beta}^{\alpha\to1}\), simply as \(F_{\beta}\).
\begin{figure}
    \centering
    \includegraphics[width=1.0\linewidth]{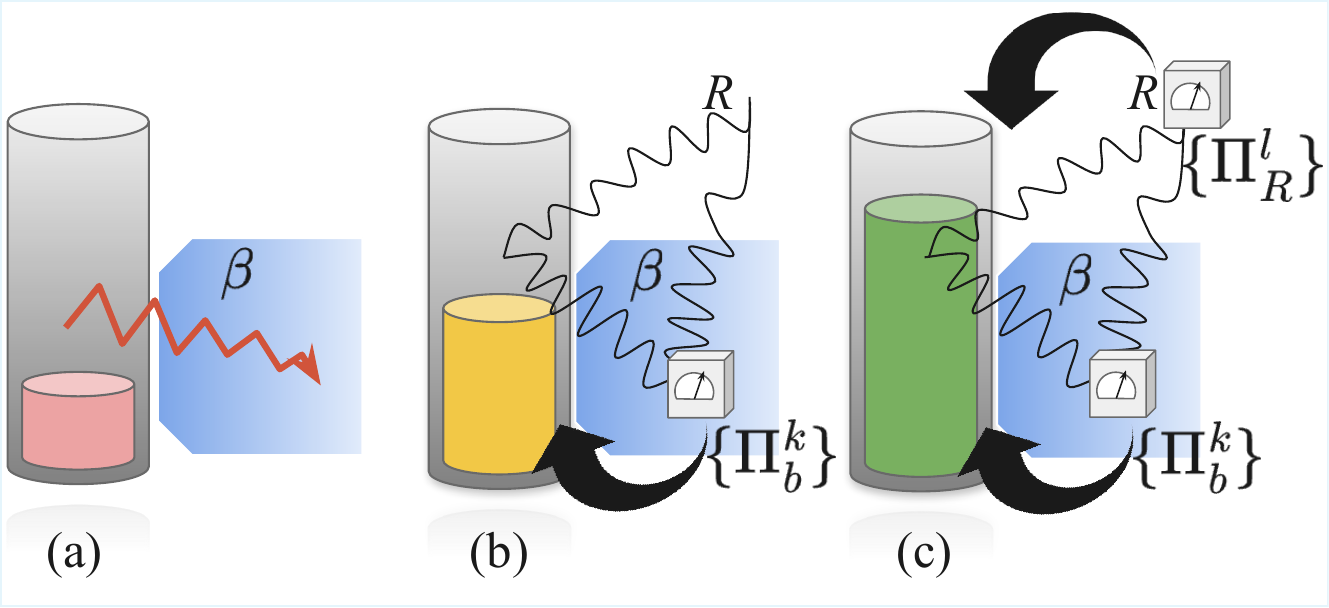}
    \caption{(Color online) \textit{A schematic diagram of charge retrieval in a quantum battery, assisted by the thermal environment}. (a) The free energetic charge stored in a quantum battery inevitably degraded due to the interaction with a thermal environment of inverse temperature \(\beta\). (b) Retrieval of the lost charge in the QB under weak assistance. Here the agent using the QB classically aided by the information extracted from the bath qudit after performing the POVM \(\{\Pi^k_b\}_k\) on it. (C) Strong retrieval of the lost charge in the QB. A suitable local POVM \(\{\Pi^k_b\otimes \Pi^l_R\}\) is performed on the bath and the reference qudit individually and informed to the battery agent. Accordingly the free-energetic charge enhances in the updated QB.}
    \label{f1}
\end{figure}

Let us now consider a positive operator valued  measurement (POVM) \(\{\Pi_b^k\}_k\) is performed on the bath quantum state. Then the resulting post-measurement battery states  denoted as \(\{\sigma_{s|k}\}_{k}\) are obtained with  probabilities $p_k$ i.e.
\begin{align}\label{e6}
\nonumber 
\sigma_{s|k}&=\frac 1{p_k} \Tr_b\left[(\mathbb{I}_s\otimes\Pi_b^k)\Tr_{R}(\sigma_{sbR})\right] ,\text{ and }\\
p_k&=\Tr\left[(\mathbb{I}_s\otimes\Pi_b^k)\Tr_{R}(\sigma_{sbR})\right],
\end{align}
We then quantify the amount of charge retrieved in the battery as 
  \[W(\rho_s,\Lambda_{\beta},\{\Pi^k_b\})=\sum_k p_k F_{\beta}(\sigma_{s|k}).\]
Moreover, we can formally quantify the maximal amount of charge retrieval under the assistance from the interacting thermal bath qudit. 
  \begin{definition}\label{d1}
      The optimal charge retrieval under weak assistance is defined as
      \[W_{weak}(\rho_s,\Lambda_{\beta})=\max_{\{\Pi^k_b\}}W(\rho_s,\Lambda_{\beta},\{\Pi^k_b\})\]
  \end{definition}
In a similar fashion, one can consider the amount of charge retrieval when both the reference and bath qudit can be locally accessed , i.e.,
\begin{align}\label{e7}
    W(\rho_s,\Lambda_{\beta},\{\Pi_b^k\otimes\Pi_R^l\})&=\sum_{k,l} p_{k,l} F_{\beta}(\sigma_{s|k,l})\\\nonumber\text{where, }p_{k,l}&=\Tr[(\mathbb{I}_s\otimes\Pi_{b}^k\otimes\Pi_{R}^l)\sigma_{sbR}],\\\nonumber\text{and }\sigma_{s|k,l}&=\frac 1{p_{k,l}}\Tr_{bR}[(\mathbb{I}_s\otimes\Pi_{b}^k\otimes\Pi_{R}^l)\sigma_{sbR}]
\end{align}
Accordingly, we can define
\begin{definition}\label{d2}
   The optimal charge retrieval under strong assistance is defined as
    \[W_{strong}(\rho_s,\Lambda_{\beta})=\max_{\{\Pi^k_b\otimes\Pi^l_R\}}W(\rho_s,\Lambda_{\beta},\{\Pi^k_b\otimes\Pi^l_R\})\].
\end{definition}
After setting the complete framework for the environment assisted charge retrieval in a quantum battery, we can now give the following proposition, which will be crucial for the subsequent results.
\begin{proposition}\label{p1}
    For any thermal operation \(\Lambda_{\beta}\) and an arbitrary input quantum state \(\rho_s\), we have the following relation:
    \[F_{\beta}(\sigma_s)\leq W_{weak}(\rho_s,\Lambda_{\beta})\leq W_{strong}(\rho_s,\Lambda_{\beta})\leq \Tr(\sigma_sH_s)\]
    where, \(\sigma_s=:\Lambda_{\beta}(\rho_s)\) and \(H_s\) is the governing system Hamiltonian, as mentioned earlier.
\end{proposition}\noindent
While the proof is intuitive, for completeness, we discuss it in Appendix \ref{pp1}.

It may be noted that the true amount of retrieved free-energetic charge, both for the weak and strong assistance, should be accounted in terms of the free energy difference between the average post-measurement states of the QB and the free energy of the thermal bath of inverse temperature \(\beta\), that is, \(F_{\beta}(\tau_{\beta})\). However, being a constant quantity independent of the initial QB states and the degree of assistance,  we rescale \(F_{\beta}(\tau_{\beta})=0\), through out our analysis for the sake of simplicity.

\textit{Main Results.}-- Before moving on to our main findings, we will first briefly recall a specific entanglement measure for the bipartite quantum systems. 

Consider a bipartite quantum state \(\rho_{AB}\) and we would like to compute the minimal number of maximally entangled states required to prepare this \(\rho_{AB}\) asymptotically. This quantity is generally characterized in terms of \textit{Entanglement of Formation} \(E_f(\rho_{AB})\) \cite{wootters1998entanglement}, the regularized version  in fact is the greatest among all possible bipartite entanglement measures, under some reasonable axioms \cite{horodecki2000limits}. In particular, for any arbitrary bipartite quantum state \(\rho_{AB}\)
\begin{align}
    E_f(\rho_{AB}):=\min_{p_i,\psi_i}\sum_ip_iE_f(\ketbra{\psi_i}{\psi_i}_{AB}),
\end{align}
where, \(\rho_{AB}=\sum_i p_i\lvert{\psi_i}\rangle \langle{\psi_i}\lvert_{AB}\) and the minimization is taken over all such decompositions of \(\rho_{AB}\). Additionally, for any pure bipartite state \(\ket{\psi}_{AB}\), all the entanglement measures, including the entanglement of formation \(E_f(\ketbra{\psi}{\psi})\), coincides with the von-Neumann entropy of the reduced marginal \(S(\rho_A)=S(\rho_B)\), where \(\rho_{A(B)}=\Tr_{B(A)}(\ketbra{\psi}{\psi}_{AB})\).

Coming back to the charge retrieval of a quantum battery under the weak assistance from the environment, we can obtain the following upper-bound  in terms of the amount of entanglement generated between the battery and the reference. 
\begin{theorem}\label{t1}
    For any quantum battery \(\rho_s\) and any nontrivial thermal operation \(\Lambda_{\beta}\),
    \[W_{weak}(\rho_s,\Lambda_{\beta})\leq \Tr(\sigma_sH_s)-\frac 1{\beta}E_f(\sigma_{sR}),\]
    where, \(\sigma_{sR}=\Tr_b(\sigma_{sbR})\). Moreover, the above inequality saturates whenever \(\rho_s\) is pure. 
\end{theorem}\noindent
The proof is deferred to Appendix \ref{pt1}. One could also consider an alternative instance of weak assistance, where instead of the interacting thermal particle, only the reference quantum system is allowed access. In Appendix \ref{pt1}, we have shown that a similar bound can also be derived for the retrieved charge \(W_{ref}(\rho_s,\Lambda_{\beta})\) in such a scenario. Note that the characterization of the reference-aided charge retrieval also corresponds to an operational instance, rather than a mere mathematical artifact. For instance, consider a situation where Alice and Bob shares a pure entangled \(\ket{\phi^+_{\beta}}_{AB}\) state with local thermal marginals as \(\sigma_A=\sigma_B=\tau_{\beta}\). Additionally, Bob possesses a QB with state \(\rho_s\), which inevitably interacts with \(\tau_{\beta}\) resulting in a charge leakage. Now, at a later time Alice could assist Bob to retrieve the lost charge back by performing suitable measurement in her local constituent and communicating the outcome.

Nevertheless, we now consider  the charge retrieval of the QB, under strong assistance. Note that Proposition \ref{p1} identifies an upper-bound for it. However, is this bound indeed achievable? In the following we will answer this question conclusively, the proof of which is detailed in Appendix \ref{pt2}.
\begin{theorem}\label{t2}
For any pure quantum battery \(\ket{\psi}_s\) and any nontrivial thermal operation \(\Lambda_{\beta}\),
\[W_{strong}(\ketbra{\psi}{\psi}_s,\Lambda_{\beta})= \Tr(\sigma_sH_s)=E_s,\]
 where, \(\sigma_s=\Lambda_{\beta}(\ketbra{\psi}{\psi}_s)\).
\end{theorem}
Theorem \ref{t1} and \ref{t2} depict a sharp difference between the amounts of strong and weakly restored charge in the quantum battery, which arguably costs the amount of the inaccessible entanglement generated between the reference system and QB. More precisely, during the thermal operation,the joint state of the system, bath and reference  becomes a pure entangled state, when the initial state of QB is pure. Under  weak assistance only a part (the bath qudit) of the pure entangled state can be accessed,thereby  restricting the optimal retrieval of charge.But in the strong assistance scenario,the access to the complete pure state lifts this bound and reaches to its the optimal value i.e. the total energy of the system.Similarly, for the initially mixed state QB,the joint system-bath-reference state is a mixed entangled state and can be purified with an additional reference system (\(R'\)). This renders the sub-optimality of both \(W_{weak}(\rho_s,\Lambda_{\beta})\) and \(W_{strong}(\rho_s,\Lambda_{\beta})\) for mixed QB. At this end, we can conclude 
\begin{corollary}\label{c1}
    For any arbitrary quantum battery \(\rho_s\) and any arbitrary thermal operation \(\Lambda_{\beta}\), the difference between the charge retrieval under strong and weak assistance 
    \[\Delta(\rho_s,\Lambda_{\beta}):=W_{strong}(\rho_s,\Lambda_{\beta})-W_{weak}(\rho_s,\Lambda_{\beta})\geq \frac1{\beta}E_f(\sigma_{sR})\]
    and the inequality saturates for \(\rho_s\) is to be a pure quantum battery.
\end{corollary}

The Corollary \ref{c1} sets a general lower-bound for the difference between the amount of  charge retrieval under strong and weak assistance from the environment. One may be further tempted to ask about the instances for which the gap will be maximized and  minimized. A trivial example for the latter is the energy preserving local unitary on the QB and bath qudit, that is, \(U_{sb}=U_{s}\otimes U_b\) where, \([U_s,H_s]=[U_b,H_b]=0\). Such a unitary will allow no correlations to be generated between \(\rho_s\) and \(\tau_{\beta}^b\) and hence, \(E_f(\sigma_{sR})=0\). Additionally, any measurement performed on the bath or the reference qudit will have no effect for retrieving the charge in the QB, that is, \(W_{strong}(\rho_s, \Lambda_{\beta})=W_{weak}(\rho_s, \Lambda_{\beta})=F(\sigma_s)\), saturating the first three inequalities of Proposition \ref{p1}. 

However, characterization of all the instances for the maximum and the minimum gap between strong and weak charge retrieval is in general very complicated. The difficulties involved are in two- folds: The proper form of the energy preserving unitary is parameterized by the structure of the individual Hamiltonian \(H_s\) and \(H_b\). On the other hand, depending upon the bipartite energy preserving unitary, various forms of multipartite entangled states can be generated between QB, bath and the reference qudit, for which \(E_f\) of the marginals are impossible to compute in general. Nevertheless, in the following we will characterize few of the general instances for the same, with non-degenerate \(H_s=H_b\) but involving no structure of input QB. 

\begin{theorem}\label{t3}
   For any qudit battery \(\rho_s\),\\
    (i) strong assistance reveals activation of charge retrieval if \(\Lambda_{\beta}\) is a thermal map of inverse temperature \(\beta\), i.e, \(\Lambda_{\beta}(\rho_s)=\tau_{\beta},~\forall \rho_s\).\\
    (ii) \(\Delta(\rho_s,\Lambda_{\beta})=0\) if the interacting thermal environment is at absolute zero temperature.  
\end{theorem}
\noindent
The proof is detailed in Appendix \ref{pt3}. Finally, we complete our analysis with an illustrative example of qubit battery in Appendix \ref{q2b} and \ref{exmpl}.

\textit{Discussions.}--- We have explored the possibilities of recycling an arbitrary dimensional quantum battery by retrieving the charge it lost due to spontaneous thermal interactions. The scenario in it's weakest variant involved a minimal classical assistance from the bath particle. However this is not the optimal amount, in general. On the other hand, when an additional quantum system that purifies the thermal bath, takes active participation in the retrieving process, the amount increases further and achieves optimality whenever the initial battery is a pure quantum state.  While our analysis involved charge as the free energy stored in a quantum state, the assistance of an additional system, to extract ergotropy from the target one has also been encountered in various contexts. For instance, the classical assistance from the ancillary system leads to the notion of daemonic ergotropy \cite{Francica2017}, which eventually identifies the presence of one-way discord in bipartite pure quantum states. 
Finally we have shown that the difference between the strong and the weakly assisted retrieved charge quantifies the amount of distant entanglement generated by the thermal operation between the quantum battery and the reference system. In this regard, a recent result is worth mentioning \cite{biswas2025quantum}, where incompatible measurements on the ancillary system causes additional free energy extraction from the thermodynamic systems. Moreover, there the steerability of quantum correlations has been identified as a sole reason behind such an advantageous free energy extraction. In our premises, alternatively, the difference between the retrieval of the lost charge, under various degrees of assistance, identifies with a fundamental entity, viz., entanglement of formation for multipartite quantum systems.

Our results can also be interpreted as an initiation towards the framework of quantum thermodynamics, involving the quantum instruments \cite{heinosaari2011mathematical, leppajarvi2021postprocessing, ghai2025instrument, sau2026demultiplexing}. In particular, the measurement outcomes on the thermal environment could be interpreted as the classical indices; while the post-measurement states of the QB as the quantum outputs corresponding to the instrument. Accordingly, our findings reveal that not all the free operations, i.e., the thermal operations, are free whenever their instrument realizations are concerned. This further paves the way to characterize the true set of free instruments, for which all the \(\alpha\)-free energies of a quantum state remain non-increasing even when conditioned over the classical outcomes.


\textit{Acknowledgments.}--- DDS, MM and TG would like to acknowledge fruitful discussions with Snehasish Roy Chowdhury. TG is supported by the Slovak Research and Development Agency through Grant no. APVV-22-0570, the Scientific Grant Agency of the Ministry of Education, Slovak Republic through Grant no.
VEGA 2/0128/24 and the Štefan Schwarz Support Fund 2025/OV1/046 by the Slovak Academy of Sciences.
\bibliography{References}

\begin{thebibliography}{72}%
\makeatletter
\providecommand \@ifxundefined [1]{%
 \@ifx{#1\undefined}
}%
\providecommand \@ifnum [1]{%
 \ifnum #1\expandafter \@firstoftwo
 \else \expandafter \@secondoftwo
 \fi
}%
\providecommand \@ifx [1]{%
 \ifx #1\expandafter \@firstoftwo
 \else \expandafter \@secondoftwo
 \fi
}%
\providecommand \natexlab [1]{#1}%
\providecommand \enquote  [1]{``#1''}%
\providecommand \bibnamefont  [1]{#1}%
\providecommand \bibfnamefont [1]{#1}%
\providecommand \citenamefont [1]{#1}%
\providecommand \href@noop [0]{\@secondoftwo}%
\providecommand \href [0]{\begingroup \@sanitize@url \@href}%
\providecommand \@href[1]{\@@startlink{#1}\@@href}%
\providecommand \@@href[1]{\endgroup#1\@@endlink}%
\providecommand \@sanitize@url [0]{\catcode `\\12\catcode `\$12\catcode `\&12\catcode `\#12\catcode `\^12\catcode `\_12\catcode `\%12\relax}%
\providecommand \@@startlink[1]{}%
\providecommand \@@endlink[0]{}%
\providecommand \url  [0]{\begingroup\@sanitize@url \@url }%
\providecommand \@url [1]{\endgroup\@href {#1}{\urlprefix }}%
\providecommand \urlprefix  [0]{URL }%
\providecommand \Eprint [0]{\href }%
\providecommand \doibase [0]{http://dx.doi.org/}%
\providecommand \selectlanguage [0]{\@gobble}%
\providecommand \bibinfo  [0]{\@secondoftwo}%
\providecommand \bibfield  [0]{\@secondoftwo}%
\providecommand \translation [1]{[#1]}%
\providecommand \BibitemOpen [0]{}%
\providecommand \bibitemStop [0]{}%
\providecommand \bibitemNoStop [0]{.\EOS\space}%
\providecommand \EOS [0]{\spacefactor3000\relax}%
\providecommand \BibitemShut  [1]{\csname bibitem#1\endcsname}%
\let\auto@bib@innerbib\@empty
\bibitem [{\citenamefont {Lenard}(1978)}]{lenard1978thermodynamical}%
  \BibitemOpen
  \bibfield  {author} {\bibinfo {author} {\bibfnamefont {Andrew}\ \bibnamefont {Lenard}},\ }\bibfield  {title} {\enquote {\bibinfo {title} {Thermodynamical proof of the gibbs formula for elementary quantum systems},}\ }\href {https://doi.org/10.1007/BF01011769} {\bibfield  {journal} {\bibinfo  {journal} {Journal of Statistical Physics}\ }\textbf {\bibinfo {volume} {19}},\ \bibinfo {pages} {575--586} (\bibinfo {year} {1978})}\BibitemShut {NoStop}%
\bibitem [{\citenamefont {Popescu}\ \emph {et~al.}(2006)\citenamefont {Popescu}, \citenamefont {Short},\ and\ \citenamefont {Winter}}]{popescu2006entanglement}%
  \BibitemOpen
  \bibfield  {author} {\bibinfo {author} {\bibfnamefont {Sandu}\ \bibnamefont {Popescu}}, \bibinfo {author} {\bibfnamefont {Anthony~J}\ \bibnamefont {Short}}, \ and\ \bibinfo {author} {\bibfnamefont {Andreas}\ \bibnamefont {Winter}},\ }\bibfield  {title} {\enquote {\bibinfo {title} {Entanglement and the foundations of statistical mechanics},}\ }\href {https://doi.org/10.1038/nphys444} {\bibfield  {journal} {\bibinfo  {journal} {Nature Physics}\ }\textbf {\bibinfo {volume} {2}},\ \bibinfo {pages} {754--758} (\bibinfo {year} {2006})}\BibitemShut {NoStop}%
\bibitem [{\citenamefont {Esposito}\ \emph {et~al.}(2009)\citenamefont {Esposito}, \citenamefont {Harbola},\ and\ \citenamefont {Mukamel}}]{esposito2009nonequilibrium}%
  \BibitemOpen
  \bibfield  {author} {\bibinfo {author} {\bibfnamefont {Massimiliano}\ \bibnamefont {Esposito}}, \bibinfo {author} {\bibfnamefont {Upendra}\ \bibnamefont {Harbola}}, \ and\ \bibinfo {author} {\bibfnamefont {Shaul}\ \bibnamefont {Mukamel}},\ }\bibfield  {title} {\enquote {\bibinfo {title} {Nonequilibrium fluctuations, fluctuation theorems, and counting statistics in quantum systems},}\ }\href {https://doi.org/10.1103/RevModPhys.81.1665} {\bibfield  {journal} {\bibinfo  {journal} {Reviews of modern physics}\ }\textbf {\bibinfo {volume} {81}},\ \bibinfo {pages} {1665--1702} (\bibinfo {year} {2009})}\BibitemShut {NoStop}%
\bibitem [{\citenamefont {Brandao}\ \emph {et~al.}(2013)\citenamefont {Brandao}, \citenamefont {Horodecki}, \citenamefont {Oppenheim}, \citenamefont {Renes},\ and\ \citenamefont {Spekkens}}]{brandao2013resource}%
  \BibitemOpen
  \bibfield  {author} {\bibinfo {author} {\bibfnamefont {Fernando~GSL}\ \bibnamefont {Brandao}}, \bibinfo {author} {\bibfnamefont {Micha{\l}}\ \bibnamefont {Horodecki}}, \bibinfo {author} {\bibfnamefont {Jonathan}\ \bibnamefont {Oppenheim}}, \bibinfo {author} {\bibfnamefont {Joseph~M}\ \bibnamefont {Renes}}, \ and\ \bibinfo {author} {\bibfnamefont {Robert~W}\ \bibnamefont {Spekkens}},\ }\bibfield  {title} {\enquote {\bibinfo {title} {Resource theory of quantum states out of thermal equilibrium},}\ }\href {https://doi.org/10.1103/PhysRevLett.111.250404} {\bibfield  {journal} {\bibinfo  {journal} {Physical Review Letters}\ }\textbf {\bibinfo {volume} {111}},\ \bibinfo {pages} {250404} (\bibinfo {year} {2013})}\BibitemShut {NoStop}%
\bibitem [{\citenamefont {Horodecki}\ and\ \citenamefont {Oppenheim}(2013)}]{horodecki2013fundamental}%
  \BibitemOpen
  \bibfield  {author} {\bibinfo {author} {\bibfnamefont {Micha{\l}}\ \bibnamefont {Horodecki}}\ and\ \bibinfo {author} {\bibfnamefont {Jonathan}\ \bibnamefont {Oppenheim}},\ }\bibfield  {title} {\enquote {\bibinfo {title} {Fundamental limitations for quantum and nanoscale thermodynamics},}\ }\href {https://doi.org/10.1038/ncomms3059} {\bibfield  {journal} {\bibinfo  {journal} {Nature communications}\ }\textbf {\bibinfo {volume} {4}},\ \bibinfo {pages} {2059} (\bibinfo {year} {2013})}\BibitemShut {NoStop}%
\bibitem [{\citenamefont {Goold}\ \emph {et~al.}(2016)\citenamefont {Goold}, \citenamefont {Huber}, \citenamefont {Riera}, \citenamefont {Del~Rio},\ and\ \citenamefont {Skrzypczyk}}]{goold2016role}%
  \BibitemOpen
  \bibfield  {author} {\bibinfo {author} {\bibfnamefont {John}\ \bibnamefont {Goold}}, \bibinfo {author} {\bibfnamefont {Marcus}\ \bibnamefont {Huber}}, \bibinfo {author} {\bibfnamefont {Arnau}\ \bibnamefont {Riera}}, \bibinfo {author} {\bibfnamefont {L{\'\i}dia}\ \bibnamefont {Del~Rio}}, \ and\ \bibinfo {author} {\bibfnamefont {Paul}\ \bibnamefont {Skrzypczyk}},\ }\bibfield  {title} {\enquote {\bibinfo {title} {The role of quantum information in thermodynamics—a topical review},}\ }\href {https://doi.org/10.1088/1751-8113/49/14/143001} {\bibfield  {journal} {\bibinfo  {journal} {Journal of Physics A: Mathematical and Theoretical}\ }\textbf {\bibinfo {volume} {49}},\ \bibinfo {pages} {143001} (\bibinfo {year} {2016})}\BibitemShut {NoStop}%
\bibitem [{\citenamefont {Linden}\ \emph {et~al.}(2010)\citenamefont {Linden}, \citenamefont {Popescu},\ and\ \citenamefont {Skrzypczyk}}]{linden2010small}%
  \BibitemOpen
  \bibfield  {author} {\bibinfo {author} {\bibfnamefont {Noah}\ \bibnamefont {Linden}}, \bibinfo {author} {\bibfnamefont {Sandu}\ \bibnamefont {Popescu}}, \ and\ \bibinfo {author} {\bibfnamefont {Paul}\ \bibnamefont {Skrzypczyk}},\ }\bibfield  {title} {\enquote {\bibinfo {title} {How small can thermal machines be? the smallest possible refrigerator},}\ }\href {https://doi.org/10.1103/PhysRevLett.105.130401} {\bibfield  {journal} {\bibinfo  {journal} {Physical Review Letters}\ }\textbf {\bibinfo {volume} {105}},\ \bibinfo {pages} {130401} (\bibinfo {year} {2010})}\BibitemShut {NoStop}%
\bibitem [{\citenamefont {Horowitz}\ and\ \citenamefont {Jacobs}(2014)}]{horowitz2014quantum}%
  \BibitemOpen
  \bibfield  {author} {\bibinfo {author} {\bibfnamefont {Jordan~M}\ \bibnamefont {Horowitz}}\ and\ \bibinfo {author} {\bibfnamefont {Kurt}\ \bibnamefont {Jacobs}},\ }\bibfield  {title} {\enquote {\bibinfo {title} {Quantum effects improve the energy efficiency of feedback control},}\ }\href {https://doi.org/10.1103/PhysRevE.89.042134} {\bibfield  {journal} {\bibinfo  {journal} {Physical Review E}\ }\textbf {\bibinfo {volume} {89}},\ \bibinfo {pages} {042134} (\bibinfo {year} {2014})}\BibitemShut {NoStop}%
\bibitem [{\citenamefont {Klatzow}\ \emph {et~al.}(2019)\citenamefont {Klatzow}, \citenamefont {Becker}, \citenamefont {Ledingham}, \citenamefont {Weinzetl}, \citenamefont {Kaczmarek}, \citenamefont {Saunders}, \citenamefont {Nunn}, \citenamefont {Walmsley}, \citenamefont {Uzdin},\ and\ \citenamefont {Poem}}]{klatzow2019experimental}%
  \BibitemOpen
  \bibfield  {author} {\bibinfo {author} {\bibfnamefont {James}\ \bibnamefont {Klatzow}}, \bibinfo {author} {\bibfnamefont {Jonas~N}\ \bibnamefont {Becker}}, \bibinfo {author} {\bibfnamefont {Patrick~M}\ \bibnamefont {Ledingham}}, \bibinfo {author} {\bibfnamefont {Christian}\ \bibnamefont {Weinzetl}}, \bibinfo {author} {\bibfnamefont {Krzysztof~T}\ \bibnamefont {Kaczmarek}}, \bibinfo {author} {\bibfnamefont {Dylan~J}\ \bibnamefont {Saunders}}, \bibinfo {author} {\bibfnamefont {Joshua}\ \bibnamefont {Nunn}}, \bibinfo {author} {\bibfnamefont {Ian~A}\ \bibnamefont {Walmsley}}, \bibinfo {author} {\bibfnamefont {Raam}\ \bibnamefont {Uzdin}}, \ and\ \bibinfo {author} {\bibfnamefont {Eilon}\ \bibnamefont {Poem}},\ }\bibfield  {title} {\enquote {\bibinfo {title} {Experimental demonstration of quantum effects in the operation of microscopic heat engines},}\ }\href {https://doi.org/10.1103/PhysRevLett.122.110601} {\bibfield  {journal} {\bibinfo  {journal} {Physical Review Letters}\ }\textbf {\bibinfo {volume} {122}},\
  \bibinfo {pages} {110601} (\bibinfo {year} {2019})}\BibitemShut {NoStop}%
\bibitem [{\citenamefont {Mitchison}(2019)}]{mitchison2019quantum}%
  \BibitemOpen
  \bibfield  {author} {\bibinfo {author} {\bibfnamefont {Mark~T}\ \bibnamefont {Mitchison}},\ }\bibfield  {title} {\enquote {\bibinfo {title} {Quantum thermal absorption machines: refrigerators, engines and clocks},}\ }\href {https://doi.org/10.1080/00107514.2019.1631555} {\bibfield  {journal} {\bibinfo  {journal} {Contemporary Physics}\ }\textbf {\bibinfo {volume} {60}},\ \bibinfo {pages} {164--187} (\bibinfo {year} {2019})}\BibitemShut {NoStop}%
\bibitem [{\citenamefont {Barra}(2022)}]{barra2022quantum}%
  \BibitemOpen
  \bibfield  {author} {\bibinfo {author} {\bibfnamefont {Felipe}\ \bibnamefont {Barra}},\ }\bibfield  {title} {\enquote {\bibinfo {title} {Quantum thermal machines: a simple scheme with realistic bath modelling},}\ }\href {https://doi.org/10.22331/qv-2022-09-26-68} {\bibfield  {journal} {\bibinfo  {journal} {Quantum Views}\ }\textbf {\bibinfo {volume} {6}},\ \bibinfo {pages} {68} (\bibinfo {year} {2022})}\BibitemShut {NoStop}%
\bibitem [{\citenamefont {Alicki}\ and\ \citenamefont {Fannes}(2013)}]{alicki2013entanglement}%
  \BibitemOpen
  \bibfield  {author} {\bibinfo {author} {\bibfnamefont {Robert}\ \bibnamefont {Alicki}}\ and\ \bibinfo {author} {\bibfnamefont {Mark}\ \bibnamefont {Fannes}},\ }\bibfield  {title} {\enquote {\bibinfo {title} {Entanglement boost for extractable work from ensembles of quantum batteries},}\ }\href {https://doi.org/10.1103/PhysRevE.87.042123} {\bibfield  {journal} {\bibinfo  {journal} {Physical Review E—Statistical, Nonlinear, and Soft Matter Physics}\ }\textbf {\bibinfo {volume} {87}},\ \bibinfo {pages} {042123} (\bibinfo {year} {2013})}\BibitemShut {NoStop}%
\bibitem [{\citenamefont {Binder}\ \emph {et~al.}(2015)\citenamefont {Binder}, \citenamefont {Vinjanampathy}, \citenamefont {Modi},\ and\ \citenamefont {Goold}}]{binder2015quantum}%
  \BibitemOpen
  \bibfield  {author} {\bibinfo {author} {\bibfnamefont {Felix}\ \bibnamefont {Binder}}, \bibinfo {author} {\bibfnamefont {Sai}\ \bibnamefont {Vinjanampathy}}, \bibinfo {author} {\bibfnamefont {Kavan}\ \bibnamefont {Modi}}, \ and\ \bibinfo {author} {\bibfnamefont {John}\ \bibnamefont {Goold}},\ }\bibfield  {title} {\enquote {\bibinfo {title} {Quantum thermodynamics of general quantum processes},}\ }\href {https://doi.org/10.1103/PhysRevE.91.032119} {\bibfield  {journal} {\bibinfo  {journal} {Physical Review E}\ }\textbf {\bibinfo {volume} {91}},\ \bibinfo {pages} {032119} (\bibinfo {year} {2015})}\BibitemShut {NoStop}%
\bibitem [{\citenamefont {Campaioli}\ \emph {et~al.}(2024)\citenamefont {Campaioli}, \citenamefont {Gherardini}, \citenamefont {Quach}, \citenamefont {Polini},\ and\ \citenamefont {Andolina}}]{campaioli2024colloquium}%
  \BibitemOpen
  \bibfield  {author} {\bibinfo {author} {\bibfnamefont {Francesco}\ \bibnamefont {Campaioli}}, \bibinfo {author} {\bibfnamefont {Stefano}\ \bibnamefont {Gherardini}}, \bibinfo {author} {\bibfnamefont {James~Q}\ \bibnamefont {Quach}}, \bibinfo {author} {\bibfnamefont {Marco}\ \bibnamefont {Polini}}, \ and\ \bibinfo {author} {\bibfnamefont {Gian~Marcello}\ \bibnamefont {Andolina}},\ }\bibfield  {title} {\enquote {\bibinfo {title} {Colloquium: quantum batteries},}\ }\href {Colloquium: quantum batteries,} {\bibfield  {journal} {\bibinfo  {journal} {Reviews of Modern Physics}\ }\textbf {\bibinfo {volume} {96}},\ \bibinfo {pages} {031001} (\bibinfo {year} {2024})}\BibitemShut {NoStop}%
\bibitem [{\citenamefont {Campaioli}\ \emph {et~al.}(2017)\citenamefont {Campaioli}, \citenamefont {Pollock}, \citenamefont {Binder}, \citenamefont {C{\'e}leri}, \citenamefont {Goold}, \citenamefont {Vinjanampathy},\ and\ \citenamefont {Modi}}]{campaioli2017enhancing}%
  \BibitemOpen
  \bibfield  {author} {\bibinfo {author} {\bibfnamefont {Francesco}\ \bibnamefont {Campaioli}}, \bibinfo {author} {\bibfnamefont {Felix~A}\ \bibnamefont {Pollock}}, \bibinfo {author} {\bibfnamefont {Felix~C}\ \bibnamefont {Binder}}, \bibinfo {author} {\bibfnamefont {Lucas}\ \bibnamefont {C{\'e}leri}}, \bibinfo {author} {\bibfnamefont {John}\ \bibnamefont {Goold}}, \bibinfo {author} {\bibfnamefont {Sai}\ \bibnamefont {Vinjanampathy}}, \ and\ \bibinfo {author} {\bibfnamefont {Kavan}\ \bibnamefont {Modi}},\ }\bibfield  {title} {\enquote {\bibinfo {title} {Enhancing the charging power of quantum batteries},}\ }\href {https://doi.org/10.1103/PhysRevLett.118.150601} {\bibfield  {journal} {\bibinfo  {journal} {Physical Review Letters}\ }\textbf {\bibinfo {volume} {118}},\ \bibinfo {pages} {150601} (\bibinfo {year} {2017})}\BibitemShut {NoStop}%
\bibitem [{\citenamefont {Ferraro}\ \emph {et~al.}(2018)\citenamefont {Ferraro}, \citenamefont {Campisi}, \citenamefont {Andolina}, \citenamefont {Pellegrini},\ and\ \citenamefont {Polini}}]{ferraro2018high}%
  \BibitemOpen
  \bibfield  {author} {\bibinfo {author} {\bibfnamefont {Dario}\ \bibnamefont {Ferraro}}, \bibinfo {author} {\bibfnamefont {Michele}\ \bibnamefont {Campisi}}, \bibinfo {author} {\bibfnamefont {Gian~Marcello}\ \bibnamefont {Andolina}}, \bibinfo {author} {\bibfnamefont {Vittorio}\ \bibnamefont {Pellegrini}}, \ and\ \bibinfo {author} {\bibfnamefont {Marco}\ \bibnamefont {Polini}},\ }\bibfield  {title} {\enquote {\bibinfo {title} {High-power collective charging of a solid-state quantum battery},}\ }\href {https://doi.org/10.1103/PhysRevLett.120.117702} {\bibfield  {journal} {\bibinfo  {journal} {Physical Review Letters}\ }\textbf {\bibinfo {volume} {120}},\ \bibinfo {pages} {117702} (\bibinfo {year} {2018})}\BibitemShut {NoStop}%
\bibitem [{\citenamefont {Barra}(2019)}]{barra2019dissipative}%
  \BibitemOpen
  \bibfield  {author} {\bibinfo {author} {\bibfnamefont {Felipe}\ \bibnamefont {Barra}},\ }\bibfield  {title} {\enquote {\bibinfo {title} {Dissipative charging of a quantum battery},}\ }\href {https://doi.org/10.1103/PhysRevLett.122.210601} {\bibfield  {journal} {\bibinfo  {journal} {Physical Review Letters}\ }\textbf {\bibinfo {volume} {122}},\ \bibinfo {pages} {210601} (\bibinfo {year} {2019})}\BibitemShut {NoStop}%
\bibitem [{\citenamefont {Crescente}\ \emph {et~al.}(2020)\citenamefont {Crescente}, \citenamefont {Carrega}, \citenamefont {Sassetti},\ and\ \citenamefont {Ferraro}}]{crescente2020ultrafast}%
  \BibitemOpen
  \bibfield  {author} {\bibinfo {author} {\bibfnamefont {Alba}\ \bibnamefont {Crescente}}, \bibinfo {author} {\bibfnamefont {Matteo}\ \bibnamefont {Carrega}}, \bibinfo {author} {\bibfnamefont {Maura}\ \bibnamefont {Sassetti}}, \ and\ \bibinfo {author} {\bibfnamefont {Dario}\ \bibnamefont {Ferraro}},\ }\bibfield  {title} {\enquote {\bibinfo {title} {Ultrafast charging in a two-photon dicke quantum battery},}\ }\href {https://doi.org/10.1103/PhysRevB.102.245407} {\bibfield  {journal} {\bibinfo  {journal} {Physical Review B}\ }\textbf {\bibinfo {volume} {102}},\ \bibinfo {pages} {245407} (\bibinfo {year} {2020})}\BibitemShut {NoStop}%
\bibitem [{\citenamefont {Ueki}\ \emph {et~al.}(2022)\citenamefont {Ueki}, \citenamefont {Kamimura}, \citenamefont {Matsuzaki}, \citenamefont {Yoshida},\ and\ \citenamefont {Tokura}}]{ueki2022quantum}%
  \BibitemOpen
  \bibfield  {author} {\bibinfo {author} {\bibfnamefont {Yudai}\ \bibnamefont {Ueki}}, \bibinfo {author} {\bibfnamefont {Shunsuke}\ \bibnamefont {Kamimura}}, \bibinfo {author} {\bibfnamefont {Yuichiro}\ \bibnamefont {Matsuzaki}}, \bibinfo {author} {\bibfnamefont {Kyo}\ \bibnamefont {Yoshida}}, \ and\ \bibinfo {author} {\bibfnamefont {Yasuhiro}\ \bibnamefont {Tokura}},\ }\bibfield  {title} {\enquote {\bibinfo {title} {Quantum battery based on superabsorption},}\ }\href {https://doi.org/10.7566/JPSJ.91.124002} {\bibfield  {journal} {\bibinfo  {journal} {Journal of the Physical Society of Japan}\ }\textbf {\bibinfo {volume} {91}},\ \bibinfo {pages} {124002} (\bibinfo {year} {2022})}\BibitemShut {NoStop}%
\bibitem [{\citenamefont {Pokhrel}\ and\ \citenamefont {Gea-Banacloche}(2025)}]{pokhrel2025large}%
  \BibitemOpen
  \bibfield  {author} {\bibinfo {author} {\bibfnamefont {Sagar}\ \bibnamefont {Pokhrel}}\ and\ \bibinfo {author} {\bibfnamefont {Julio}\ \bibnamefont {Gea-Banacloche}},\ }\bibfield  {title} {\enquote {\bibinfo {title} {Large collective power enhancement in dissipative charging of a quantum battery},}\ }\href {https://doi.org/10.1103/PhysRevLett.134.130401} {\bibfield  {journal} {\bibinfo  {journal} {Physical Review Letters}\ }\textbf {\bibinfo {volume} {134}},\ \bibinfo {pages} {130401} (\bibinfo {year} {2025})}\BibitemShut {NoStop}%
\bibitem [{\citenamefont {Shastri}\ \emph {et~al.}(2025)\citenamefont {Shastri}, \citenamefont {Jiang}, \citenamefont {Xu}, \citenamefont {Prasanna~Venkatesh},\ and\ \citenamefont {Watanabe}}]{shastri2025dephasing}%
  \BibitemOpen
  \bibfield  {author} {\bibinfo {author} {\bibfnamefont {Rahul}\ \bibnamefont {Shastri}}, \bibinfo {author} {\bibfnamefont {Chao}\ \bibnamefont {Jiang}}, \bibinfo {author} {\bibfnamefont {Guo-Hua}\ \bibnamefont {Xu}}, \bibinfo {author} {\bibfnamefont {B}~\bibnamefont {Prasanna~Venkatesh}}, \ and\ \bibinfo {author} {\bibfnamefont {Gentaro}\ \bibnamefont {Watanabe}},\ }\bibfield  {title} {\enquote {\bibinfo {title} {Dephasing enabled fast charging of quantum batteries},}\ }\href {https://doi.org/10.1038/s41534-025-00959-5} {\bibfield  {journal} {\bibinfo  {journal} {npj Quantum Information}\ }\textbf {\bibinfo {volume} {11}},\ \bibinfo {pages} {9} (\bibinfo {year} {2025})}\BibitemShut {NoStop}%
\bibitem [{\citenamefont {Andolina}\ \emph {et~al.}(2019)\citenamefont {Andolina}, \citenamefont {Keck}, \citenamefont {Mari}, \citenamefont {Campisi}, \citenamefont {Giovannetti},\ and\ \citenamefont {Polini}}]{adolina2019extractable}%
  \BibitemOpen
  \bibfield  {author} {\bibinfo {author} {\bibfnamefont {Gian~Marcello}\ \bibnamefont {Andolina}}, \bibinfo {author} {\bibfnamefont {Maximilian}\ \bibnamefont {Keck}}, \bibinfo {author} {\bibfnamefont {Andrea}\ \bibnamefont {Mari}}, \bibinfo {author} {\bibfnamefont {Michele}\ \bibnamefont {Campisi}}, \bibinfo {author} {\bibfnamefont {Vittorio}\ \bibnamefont {Giovannetti}}, \ and\ \bibinfo {author} {\bibfnamefont {Marco}\ \bibnamefont {Polini}},\ }\bibfield  {title} {\enquote {\bibinfo {title} {Extractable work, the role of correlations, and asymptotic freedom in quantum batteries},}\ }\href {https://doi.org/10.1103/PhysRevLett.122.047702} {\bibfield  {journal} {\bibinfo  {journal} {Physical Review Letters}\ }\textbf {\bibinfo {volume} {122}},\ \bibinfo {pages} {047702} (\bibinfo {year} {2019})}\BibitemShut {NoStop}%
\bibitem [{\citenamefont {Shi}\ \emph {et~al.}(2022)\citenamefont {Shi}, \citenamefont {Ding}, \citenamefont {Wan}, \citenamefont {Wang},\ and\ \citenamefont {Yang}}]{shi2022entanglement}%
  \BibitemOpen
  \bibfield  {author} {\bibinfo {author} {\bibfnamefont {Hai-Long}\ \bibnamefont {Shi}}, \bibinfo {author} {\bibfnamefont {Shu}\ \bibnamefont {Ding}}, \bibinfo {author} {\bibfnamefont {Qing-Kun}\ \bibnamefont {Wan}}, \bibinfo {author} {\bibfnamefont {Xiao-Hui}\ \bibnamefont {Wang}}, \ and\ \bibinfo {author} {\bibfnamefont {Wen-Li}\ \bibnamefont {Yang}},\ }\bibfield  {title} {\enquote {\bibinfo {title} {Entanglement, coherence, and extractable work in quantum batteries},}\ }\href {https://doi.org/10.1103/PhysRevLett.129.130602} {\bibfield  {journal} {\bibinfo  {journal} {Physical Review Letters}\ }\textbf {\bibinfo {volume} {129}},\ \bibinfo {pages} {130602} (\bibinfo {year} {2022})}\BibitemShut {NoStop}%
\bibitem [{\citenamefont {Lu}\ \emph {et~al.}(2025)\citenamefont {Lu}, \citenamefont {Tian}, \citenamefont {L{\"u}},\ and\ \citenamefont {Shang}}]{lu2025topological}%
  \BibitemOpen
  \bibfield  {author} {\bibinfo {author} {\bibfnamefont {Zhi-Guang}\ \bibnamefont {Lu}}, \bibinfo {author} {\bibfnamefont {Guoqing}\ \bibnamefont {Tian}}, \bibinfo {author} {\bibfnamefont {Xin-You}\ \bibnamefont {L{\"u}}}, \ and\ \bibinfo {author} {\bibfnamefont {Cheng}\ \bibnamefont {Shang}},\ }\bibfield  {title} {\enquote {\bibinfo {title} {Topological quantum batteries},}\ }\href {https://doi.org/10.1103/PhysRevLett.134.180401} {\bibfield  {journal} {\bibinfo  {journal} {Physical Review Letters}\ }\textbf {\bibinfo {volume} {134}},\ \bibinfo {pages} {180401} (\bibinfo {year} {2025})}\BibitemShut {NoStop}%
\bibitem [{\citenamefont {Friis}\ and\ \citenamefont {Huber}(2018)}]{friis2018precision}%
  \BibitemOpen
  \bibfield  {author} {\bibinfo {author} {\bibfnamefont {Nicolai}\ \bibnamefont {Friis}}\ and\ \bibinfo {author} {\bibfnamefont {Marcus}\ \bibnamefont {Huber}},\ }\bibfield  {title} {\enquote {\bibinfo {title} {Precision and work fluctuations in gaussian battery charging},}\ }\href {https://doi.org/10.22331/q-2018-04-23-61} {\bibfield  {journal} {\bibinfo  {journal} {Quantum}\ }\textbf {\bibinfo {volume} {2}},\ \bibinfo {pages} {61} (\bibinfo {year} {2018})}\BibitemShut {NoStop}%
\bibitem [{\citenamefont {Juli{\`a}-Farr{\'e}}\ \emph {et~al.}(2020)\citenamefont {Juli{\`a}-Farr{\'e}}, \citenamefont {Salamon}, \citenamefont {Riera}, \citenamefont {Bera},\ and\ \citenamefont {Lewenstein}}]{julia2020bounds}%
  \BibitemOpen
  \bibfield  {author} {\bibinfo {author} {\bibfnamefont {Sergi}\ \bibnamefont {Juli{\`a}-Farr{\'e}}}, \bibinfo {author} {\bibfnamefont {Tymoteusz}\ \bibnamefont {Salamon}}, \bibinfo {author} {\bibfnamefont {Arnau}\ \bibnamefont {Riera}}, \bibinfo {author} {\bibfnamefont {Manabendra~N}\ \bibnamefont {Bera}}, \ and\ \bibinfo {author} {\bibfnamefont {Maciej}\ \bibnamefont {Lewenstein}},\ }\bibfield  {title} {\enquote {\bibinfo {title} {Bounds on the capacity and power of quantum batteries},}\ }\href {https://doi.org/10.1103/PhysRevResearch.2.023113} {\bibfield  {journal} {\bibinfo  {journal} {Physical Review Research}\ }\textbf {\bibinfo {volume} {2}},\ \bibinfo {pages} {023113} (\bibinfo {year} {2020})}\BibitemShut {NoStop}%
\bibitem [{\citenamefont {Peng}\ \emph {et~al.}(2021)\citenamefont {Peng}, \citenamefont {He}, \citenamefont {Chesi}, \citenamefont {Lin},\ and\ \citenamefont {Guan}}]{peng2021lower}%
  \BibitemOpen
  \bibfield  {author} {\bibinfo {author} {\bibfnamefont {Li}~\bibnamefont {Peng}}, \bibinfo {author} {\bibfnamefont {Wen-Bin}\ \bibnamefont {He}}, \bibinfo {author} {\bibfnamefont {Stefano}\ \bibnamefont {Chesi}}, \bibinfo {author} {\bibfnamefont {Hai-Qing}\ \bibnamefont {Lin}}, \ and\ \bibinfo {author} {\bibfnamefont {Xi-Wen}\ \bibnamefont {Guan}},\ }\bibfield  {title} {\enquote {\bibinfo {title} {Lower and upper bounds of quantum battery power in multiple central spin systems},}\ }\href {https://doi.org/10.1103/PhysRevA.103.052220} {\bibfield  {journal} {\bibinfo  {journal} {Physical Review A}\ }\textbf {\bibinfo {volume} {103}},\ \bibinfo {pages} {052220} (\bibinfo {year} {2021})}\BibitemShut {NoStop}%
\bibitem [{\citenamefont {Tirone}\ \emph {et~al.}(2024)\citenamefont {Tirone}, \citenamefont {Salvia}, \citenamefont {Chessa},\ and\ \citenamefont {Giovannetti}}]{tirone2024quantum}%
  \BibitemOpen
  \bibfield  {author} {\bibinfo {author} {\bibfnamefont {Salvatore}\ \bibnamefont {Tirone}}, \bibinfo {author} {\bibfnamefont {Raffaele}\ \bibnamefont {Salvia}}, \bibinfo {author} {\bibfnamefont {Stefano}\ \bibnamefont {Chessa}}, \ and\ \bibinfo {author} {\bibfnamefont {Vittorio}\ \bibnamefont {Giovannetti}},\ }\bibfield  {title} {\enquote {\bibinfo {title} {Quantum work capacitances: Ultimate limits for energy extraction on noisy quantum batteries},}\ }\href {https://doi.org/10.21468/SciPostPhys.17.2.041} {\bibfield  {journal} {\bibinfo  {journal} {SciPost Physics}\ }\textbf {\bibinfo {volume} {17}},\ \bibinfo {pages} {041} (\bibinfo {year} {2024})}\BibitemShut {NoStop}%
\bibitem [{\citenamefont {Bai}\ \emph {et~al.}(2024)\citenamefont {Bai}, \citenamefont {Gong},\ and\ \citenamefont {Li}}]{bai2024change}%
  \BibitemOpen
  \bibfield  {author} {\bibinfo {author} {\bibfnamefont {Guoji}\ \bibnamefont {Bai}}, \bibinfo {author} {\bibfnamefont {Helin}\ \bibnamefont {Gong}}, \ and\ \bibinfo {author} {\bibfnamefont {Bo}~\bibnamefont {Li}},\ }\bibfield  {title} {\enquote {\bibinfo {title} {The change in quantum battery capacity under local pauli noise},}\ }\href {https://doi.org/10.1140/epjp/s13360-024-05847-z} {\bibfield  {journal} {\bibinfo  {journal} {The European Physical Journal Plus}\ }\textbf {\bibinfo {volume} {139}},\ \bibinfo {pages} {1--10} (\bibinfo {year} {2024})}\BibitemShut {NoStop}%
\bibitem [{\citenamefont {Liu}\ \emph {et~al.}(2025)\citenamefont {Liu}, \citenamefont {Tian},\ and\ \citenamefont {Wang}}]{liu2506open}%
  \BibitemOpen
  \bibfield  {author} {\bibinfo {author} {\bibfnamefont {Xiaofang}\ \bibnamefont {Liu}}, \bibinfo {author} {\bibfnamefont {Zehua}\ \bibnamefont {Tian}}, \ and\ \bibinfo {author} {\bibfnamefont {Jieci}\ \bibnamefont {Wang}},\ }\bibfield  {title} {\enquote {\bibinfo {title} {Open quantum battery in three-dimensional rotating black hole spacetime},}\ }\href {https://doi.org/10.48550/arXiv.2506.07568} {\bibfield  {journal} {\bibinfo  {journal} {arXiv preprint arXiv:2506.07568}\ } (\bibinfo {year} {2025})}\BibitemShut {NoStop}%
\bibitem [{\citenamefont {Malavazi}\ \emph {et~al.}(2025)\citenamefont {Malavazi}, \citenamefont {Sagar}, \citenamefont {Ahmadi},\ and\ \citenamefont {Dieguez}}]{malavazi2025two}%
  \BibitemOpen
  \bibfield  {author} {\bibinfo {author} {\bibfnamefont {Andr{\'e}~HA}\ \bibnamefont {Malavazi}}, \bibinfo {author} {\bibfnamefont {Rishav}\ \bibnamefont {Sagar}}, \bibinfo {author} {\bibfnamefont {Borhan}\ \bibnamefont {Ahmadi}}, \ and\ \bibinfo {author} {\bibfnamefont {Pedro~R}\ \bibnamefont {Dieguez}},\ }\bibfield  {title} {\enquote {\bibinfo {title} {Two-time weak-measurement protocol for ergotropy protection in open quantum batteries},}\ }\href {https://doi.org/10.1103/bv4w-jr6q} {\bibfield  {journal} {\bibinfo  {journal} {PRX Energy}\ }\textbf {\bibinfo {volume} {4}},\ \bibinfo {pages} {023011} (\bibinfo {year} {2025})}\BibitemShut {NoStop}%
\bibitem [{\citenamefont {Skrzypczyk}\ \emph {et~al.}(2014)\citenamefont {Skrzypczyk}, \citenamefont {Short},\ and\ \citenamefont {Popescu}}]{skrzypczyk2014work}%
  \BibitemOpen
  \bibfield  {author} {\bibinfo {author} {\bibfnamefont {Paul}\ \bibnamefont {Skrzypczyk}}, \bibinfo {author} {\bibfnamefont {Anthony~J}\ \bibnamefont {Short}}, \ and\ \bibinfo {author} {\bibfnamefont {Sandu}\ \bibnamefont {Popescu}},\ }\bibfield  {title} {\enquote {\bibinfo {title} {Work extraction and thermodynamics for individual quantum systems},}\ }\href {https://doi.org/10.1038/ncomms5185} {\bibfield  {journal} {\bibinfo  {journal} {Nature communications}\ }\textbf {\bibinfo {volume} {5}},\ \bibinfo {pages} {4185} (\bibinfo {year} {2014})}\BibitemShut {NoStop}%
\bibitem [{\citenamefont {Morris}\ \emph {et~al.}(2019)\citenamefont {Morris}, \citenamefont {Lami},\ and\ \citenamefont {Adesso}}]{morris2019assisted}%
  \BibitemOpen
  \bibfield  {author} {\bibinfo {author} {\bibfnamefont {Benjamin}\ \bibnamefont {Morris}}, \bibinfo {author} {\bibfnamefont {Ludovico}\ \bibnamefont {Lami}}, \ and\ \bibinfo {author} {\bibfnamefont {Gerardo}\ \bibnamefont {Adesso}},\ }\bibfield  {title} {\enquote {\bibinfo {title} {Assisted work distillation},}\ }\href {https://doi.org/10.1103/PhysRevLett.122.130601} {\bibfield  {journal} {\bibinfo  {journal} {Physical Review Letters}\ }\textbf {\bibinfo {volume} {122}},\ \bibinfo {pages} {130601} (\bibinfo {year} {2019})}\BibitemShut {NoStop}%
\bibitem [{\citenamefont {Guha}\ \emph {et~al.}(2020)\citenamefont {Guha}, \citenamefont {Alimuddin},\ and\ \citenamefont {Parashar}}]{guha2020thermodynamic}%
  \BibitemOpen
  \bibfield  {author} {\bibinfo {author} {\bibfnamefont {Tamal}\ \bibnamefont {Guha}}, \bibinfo {author} {\bibfnamefont {Mir}\ \bibnamefont {Alimuddin}}, \ and\ \bibinfo {author} {\bibfnamefont {Preeti}\ \bibnamefont {Parashar}},\ }\bibfield  {title} {\enquote {\bibinfo {title} {Thermodynamic advancement in the causally inseparable occurrence of thermal maps},}\ }\href {https://doi.org/10.1103/PhysRevA.102.032215} {\bibfield  {journal} {\bibinfo  {journal} {Physical Review A}\ }\textbf {\bibinfo {volume} {102}},\ \bibinfo {pages} {032215} (\bibinfo {year} {2020})}\BibitemShut {NoStop}%
\bibitem [{\citenamefont {{Allahverdyan, A. E.}}\ \emph {et~al.}(2004)\citenamefont {{Allahverdyan, A. E.}}, \citenamefont {{Balian, R.}},\ and\ \citenamefont {{Nieuwenhuizen, Th. M.}}}]{Allahverdyan}%
  \BibitemOpen
  \bibfield  {author} {\bibinfo {author} {\bibnamefont {{Allahverdyan, A. E.}}}, \bibinfo {author} {\bibnamefont {{Balian, R.}}}, \ and\ \bibinfo {author} {\bibnamefont {{Nieuwenhuizen, Th. M.}}},\ }\bibfield  {title} {\enquote {\bibinfo {title} {Maximal work extraction from finite quantum systems},}\ }\href {https://doi.org/10.1209/epl/i2004-10101-2} {\bibfield  {journal} {\bibinfo  {journal} {Europhys. Lett.}\ }\textbf {\bibinfo {volume} {67}},\ \bibinfo {pages} {565--571} (\bibinfo {year} {2004})}\BibitemShut {NoStop}%
\bibitem [{\citenamefont {Castellano}\ \emph {et~al.}(2024)\citenamefont {Castellano}, \citenamefont {Farina}, \citenamefont {Giovannetti},\ and\ \citenamefont {Acin}}]{castellano2024extended}%
  \BibitemOpen
  \bibfield  {author} {\bibinfo {author} {\bibfnamefont {Riccardo}\ \bibnamefont {Castellano}}, \bibinfo {author} {\bibfnamefont {Donato}\ \bibnamefont {Farina}}, \bibinfo {author} {\bibfnamefont {Vittorio}\ \bibnamefont {Giovannetti}}, \ and\ \bibinfo {author} {\bibfnamefont {Antonio}\ \bibnamefont {Acin}},\ }\bibfield  {title} {\enquote {\bibinfo {title} {Extended local ergotropy},}\ }\href {https://doi.org/10.1103/PhysRevLett.133.150402} {\bibfield  {journal} {\bibinfo  {journal} {Physical Review Letters}\ }\textbf {\bibinfo {volume} {133}},\ \bibinfo {pages} {150402} (\bibinfo {year} {2024})}\BibitemShut {NoStop}%
\bibitem [{\citenamefont {Halder}\ \emph {et~al.}(2024)\citenamefont {Halder}, \citenamefont {Ghosh}, \citenamefont {Roy},\ and\ \citenamefont {Guha}}]{halder2024operational}%
  \BibitemOpen
  \bibfield  {author} {\bibinfo {author} {\bibfnamefont {Pritam}\ \bibnamefont {Halder}}, \bibinfo {author} {\bibfnamefont {Srijon}\ \bibnamefont {Ghosh}}, \bibinfo {author} {\bibfnamefont {Saptarshi}\ \bibnamefont {Roy}}, \ and\ \bibinfo {author} {\bibfnamefont {Tamal}\ \bibnamefont {Guha}},\ }\bibfield  {title} {\enquote {\bibinfo {title} {Operational ergotropy: suboptimality of the geodesic drive},}\ }\href {https://doi.org/10.48550/arXiv.2403.05956} {\bibfield  {journal} {\bibinfo  {journal} {arXiv preprint arXiv:2403.05956}\ } (\bibinfo {year} {2024})}\BibitemShut {NoStop}%
\bibitem [{\citenamefont {Simonov}\ \emph {et~al.}(2022)\citenamefont {Simonov}, \citenamefont {Francica}, \citenamefont {Guarnieri},\ and\ \citenamefont {Paternostro}}]{simonov2022work}%
  \BibitemOpen
  \bibfield  {author} {\bibinfo {author} {\bibfnamefont {Kyrylo}\ \bibnamefont {Simonov}}, \bibinfo {author} {\bibfnamefont {Gianluca}\ \bibnamefont {Francica}}, \bibinfo {author} {\bibfnamefont {Giacomo}\ \bibnamefont {Guarnieri}}, \ and\ \bibinfo {author} {\bibfnamefont {Mauro}\ \bibnamefont {Paternostro}},\ }\bibfield  {title} {\enquote {\bibinfo {title} {Work extraction from coherently activated maps via quantum switch},}\ }\href {https://doi.org/10.1103/PhysRevA.105.032217} {\bibfield  {journal} {\bibinfo  {journal} {Physical Review A}\ }\textbf {\bibinfo {volume} {105}},\ \bibinfo {pages} {032217} (\bibinfo {year} {2022})}\BibitemShut {NoStop}%
\bibitem [{\citenamefont {Simonov}\ \emph {et~al.}(2025)\citenamefont {Simonov}, \citenamefont {Roy}, \citenamefont {Guha}, \citenamefont {Zimbor{\'a}s},\ and\ \citenamefont {Chiribella}}]{simonov2025activation}%
  \BibitemOpen
  \bibfield  {author} {\bibinfo {author} {\bibfnamefont {Kyrylo}\ \bibnamefont {Simonov}}, \bibinfo {author} {\bibfnamefont {Saptarshi}\ \bibnamefont {Roy}}, \bibinfo {author} {\bibfnamefont {Tamal}\ \bibnamefont {Guha}}, \bibinfo {author} {\bibfnamefont {Zolt{\'a}n}\ \bibnamefont {Zimbor{\'a}s}}, \ and\ \bibinfo {author} {\bibfnamefont {Giulio}\ \bibnamefont {Chiribella}},\ }\bibfield  {title} {\enquote {\bibinfo {title} {Activation of thermal states by coherently controlled thermalization processes},}\ }\href {https://iopscience.iop.org/article/10.1088/1367-2630/ade5c4/meta} {\bibfield  {journal} {\bibinfo  {journal} {New Journal of Physics}\ }\textbf {\bibinfo {volume} {27}} (\bibinfo {year} {2025})}\BibitemShut {NoStop}%
\bibitem [{\citenamefont {Gregoratti}\ and\ \citenamefont {Werner}(2003)}]{gregoratti2003quantum}%
  \BibitemOpen
  \bibfield  {author} {\bibinfo {author} {\bibfnamefont {Matteo}\ \bibnamefont {Gregoratti}}\ and\ \bibinfo {author} {\bibfnamefont {Reinhard~F}\ \bibnamefont {Werner}},\ }\bibfield  {title} {\enquote {\bibinfo {title} {Quantum lost and found},}\ }\href {https://doi.org/10.1080/09500340308234541} {\bibfield  {journal} {\bibinfo  {journal} {Journal of Modern Optics}\ }\textbf {\bibinfo {volume} {50}},\ \bibinfo {pages} {915--933} (\bibinfo {year} {2003})}\BibitemShut {NoStop}%
\bibitem [{\citenamefont {Hayden}\ and\ \citenamefont {King}(2005)}]{hayden2005correcting}%
  \BibitemOpen
  \bibfield  {author} {\bibinfo {author} {\bibfnamefont {Patrick}\ \bibnamefont {Hayden}}\ and\ \bibinfo {author} {\bibfnamefont {Christopher}\ \bibnamefont {King}},\ }\bibfield  {title} {\enquote {\bibinfo {title} {Correcting quantum channels by measuring the environment},}\ }\href {https://dl.acm.org/doi/10.5555/2011626.2011632} {\bibfield  {journal} {\bibinfo  {journal} {Quant. Info. Comp.}\ }\textbf {\bibinfo {volume} {5}},\ \bibinfo {pages} {156} (\bibinfo {year} {2005})}\BibitemShut {NoStop}%
\bibitem [{\citenamefont {Karumanchi}\ \emph {et~al.}(2016{\natexlab{a}})\citenamefont {Karumanchi}, \citenamefont {Mancini}, \citenamefont {Winter},\ and\ \citenamefont {Yang}}]{karumanchi2016classical}%
  \BibitemOpen
  \bibfield  {author} {\bibinfo {author} {\bibfnamefont {Siddharth}\ \bibnamefont {Karumanchi}}, \bibinfo {author} {\bibfnamefont {Stefano}\ \bibnamefont {Mancini}}, \bibinfo {author} {\bibfnamefont {Andreas}\ \bibnamefont {Winter}}, \ and\ \bibinfo {author} {\bibfnamefont {Dong}\ \bibnamefont {Yang}},\ }\bibfield  {title} {\enquote {\bibinfo {title} {Classical capacities of quantum channels with environment assistance},}\ }\href {https://doi.org/10.1134/S0032946016030029} {\bibfield  {journal} {\bibinfo  {journal} {Problems of Information Transmission}\ }\textbf {\bibinfo {volume} {52}},\ \bibinfo {pages} {214--238} (\bibinfo {year} {2016}{\natexlab{a}})}\BibitemShut {NoStop}%
\bibitem [{\citenamefont {Karumanchi}\ \emph {et~al.}(2016{\natexlab{b}})\citenamefont {Karumanchi}, \citenamefont {Mancini}, \citenamefont {Winter},\ and\ \citenamefont {Yang}}]{karumanchi2016quantum}%
  \BibitemOpen
  \bibfield  {author} {\bibinfo {author} {\bibfnamefont {Siddharth}\ \bibnamefont {Karumanchi}}, \bibinfo {author} {\bibfnamefont {Stefano}\ \bibnamefont {Mancini}}, \bibinfo {author} {\bibfnamefont {Andreas}\ \bibnamefont {Winter}}, \ and\ \bibinfo {author} {\bibfnamefont {Dong}\ \bibnamefont {Yang}},\ }\bibfield  {title} {\enquote {\bibinfo {title} {Quantum channel capacities with passive environment assistance},}\ }\href {https://doi.org/10.1109/TIT.2016.2522192} {\bibfield  {journal} {\bibinfo  {journal} {IEEE Transactions on Information Theory}\ }\textbf {\bibinfo {volume} {62}},\ \bibinfo {pages} {1733--1747} (\bibinfo {year} {2016}{\natexlab{b}})}\BibitemShut {NoStop}%
\bibitem [{\citenamefont {Pirandola}\ \emph {et~al.}(2021)\citenamefont {Pirandola}, \citenamefont {Ottaviani}, \citenamefont {Jacobsen}, \citenamefont {Spedalieri}, \citenamefont {Braunstein}, \citenamefont {Gehring},\ and\ \citenamefont {Andersen}}]{pirandola2021environment}%
  \BibitemOpen
  \bibfield  {author} {\bibinfo {author} {\bibfnamefont {Stefano}\ \bibnamefont {Pirandola}}, \bibinfo {author} {\bibfnamefont {Carlo}\ \bibnamefont {Ottaviani}}, \bibinfo {author} {\bibfnamefont {Christian~S}\ \bibnamefont {Jacobsen}}, \bibinfo {author} {\bibfnamefont {Gaetana}\ \bibnamefont {Spedalieri}}, \bibinfo {author} {\bibfnamefont {Samuel~L}\ \bibnamefont {Braunstein}}, \bibinfo {author} {\bibfnamefont {Tobias}\ \bibnamefont {Gehring}}, \ and\ \bibinfo {author} {\bibfnamefont {Ulrik~L}\ \bibnamefont {Andersen}},\ }\bibfield  {title} {\enquote {\bibinfo {title} {Environment-assisted bosonic quantum communications},}\ }\href {https://doi.org/10.1038/s41534-021-00413-2} {\bibfield  {journal} {\bibinfo  {journal} {npj Quantum Information}\ }\textbf {\bibinfo {volume} {7}},\ \bibinfo {pages} {77} (\bibinfo {year} {2021})}\BibitemShut {NoStop}%
\bibitem [{\citenamefont {Oskouei}\ \emph {et~al.}(2021)\citenamefont {Oskouei}, \citenamefont {Mancini},\ and\ \citenamefont {Winter}}]{oskouei2021capacities}%
  \BibitemOpen
  \bibfield  {author} {\bibinfo {author} {\bibfnamefont {Samad~Khabbazi}\ \bibnamefont {Oskouei}}, \bibinfo {author} {\bibfnamefont {Stefano}\ \bibnamefont {Mancini}}, \ and\ \bibinfo {author} {\bibfnamefont {Andreas}\ \bibnamefont {Winter}},\ }\bibfield  {title} {\enquote {\bibinfo {title} {Capacities of gaussian quantum channels with passive environment assistance},}\ }\href {https://doi.org/10.1109/TIT.2021.3122150} {\bibfield  {journal} {\bibinfo  {journal} {IEEE Transactions on Information Theory}\ }\textbf {\bibinfo {volume} {68}},\ \bibinfo {pages} {339--358} (\bibinfo {year} {2021})}\BibitemShut {NoStop}%
\bibitem [{\citenamefont {Harraz}\ \emph {et~al.}(2022)\citenamefont {Harraz}, \citenamefont {Cong},\ and\ \citenamefont {Nieto}}]{harraz2022quantum}%
  \BibitemOpen
  \bibfield  {author} {\bibinfo {author} {\bibfnamefont {Sajede}\ \bibnamefont {Harraz}}, \bibinfo {author} {\bibfnamefont {Shuang}\ \bibnamefont {Cong}}, \ and\ \bibinfo {author} {\bibfnamefont {Juan~J}\ \bibnamefont {Nieto}},\ }\bibfield  {title} {\enquote {\bibinfo {title} {Quantum state recovery via environment-assisted measurement and weak measurement},}\ }\href {https://doi.org/10.1007/s10773-022-05055-4} {\bibfield  {journal} {\bibinfo  {journal} {International Journal of Theoretical Physics}\ }\textbf {\bibinfo {volume} {61}},\ \bibinfo {pages} {140} (\bibinfo {year} {2022})}\BibitemShut {NoStop}%
\bibitem [{\citenamefont {Harraz}\ \emph {et~al.}(2023)\citenamefont {Harraz}, \citenamefont {Zhang},\ and\ \citenamefont {Cong}}]{harraz2023high}%
  \BibitemOpen
  \bibfield  {author} {\bibinfo {author} {\bibfnamefont {Sajede}\ \bibnamefont {Harraz}}, \bibinfo {author} {\bibfnamefont {Jiao-Yang}\ \bibnamefont {Zhang}}, \ and\ \bibinfo {author} {\bibfnamefont {Shuang}\ \bibnamefont {Cong}},\ }\bibfield  {title} {\enquote {\bibinfo {title} {High-fidelity quantum teleportation through noisy channels via weak measurement and environment-assisted measurement},}\ }\href {https://doi.org/10.48550/arXiv.2206.14463} {\bibfield  {journal} {\bibinfo  {journal} {Results in Physics}\ }\textbf {\bibinfo {volume} {55}},\ \bibinfo {pages} {107164} (\bibinfo {year} {2023})}\BibitemShut {NoStop}%
\bibitem [{\citenamefont {Chowdhury}\ \emph {et~al.}(2025)\citenamefont {Chowdhury}, \citenamefont {Saha}, \citenamefont {Ghosh}, \citenamefont {Adhikary},\ and\ \citenamefont {Guha}}]{chowdhury2025minimal}%
  \BibitemOpen
  \bibfield  {author} {\bibinfo {author} {\bibfnamefont {Snehasish~Roy}\ \bibnamefont {Chowdhury}}, \bibinfo {author} {\bibfnamefont {Sutapa}\ \bibnamefont {Saha}}, \bibinfo {author} {\bibfnamefont {Subhendu~B}\ \bibnamefont {Ghosh}}, \bibinfo {author} {\bibfnamefont {Ranendu}\ \bibnamefont {Adhikary}}, \ and\ \bibinfo {author} {\bibfnamefont {Tamal}\ \bibnamefont {Guha}},\ }\bibfield  {title} {\enquote {\bibinfo {title} {Minimal help, maximal gain: Environmental assistance unlocks encoding strength},}\ }\href {https://doi.org/10.48550/arXiv.2509.09340} {\bibfield  {journal} {\bibinfo  {journal} {arXiv preprint arXiv:2509.09340}\ } (\bibinfo {year} {2025})}\BibitemShut {NoStop}%
\bibitem [{\citenamefont {Macchiavello}\ and\ \citenamefont {Palma}(2002)}]{macchiavello2002entanglement}%
  \BibitemOpen
  \bibfield  {author} {\bibinfo {author} {\bibfnamefont {Chiara}\ \bibnamefont {Macchiavello}}\ and\ \bibinfo {author} {\bibfnamefont {G~Massimo}\ \bibnamefont {Palma}},\ }\bibfield  {title} {\enquote {\bibinfo {title} {Entanglement-enhanced information transmission over a quantum channel with correlated noise},}\ }\href {https://doi.org/10.1103/PhysRevA.65.050301} {\bibfield  {journal} {\bibinfo  {journal} {Physical Review A}\ }\textbf {\bibinfo {volume} {65}},\ \bibinfo {pages} {050301} (\bibinfo {year} {2002})}\BibitemShut {NoStop}%
\bibitem [{\citenamefont {Bowen}\ and\ \citenamefont {Mancini}(2004)}]{bowen2004quantum}%
  \BibitemOpen
  \bibfield  {author} {\bibinfo {author} {\bibfnamefont {Garry}\ \bibnamefont {Bowen}}\ and\ \bibinfo {author} {\bibfnamefont {Stefano}\ \bibnamefont {Mancini}},\ }\bibfield  {title} {\enquote {\bibinfo {title} {Quantum channels with a finite memory},}\ }\href {https://doi.org/10.1103/PhysRevA.69.012306} {\bibfield  {journal} {\bibinfo  {journal} {Physical Review A}\ }\textbf {\bibinfo {volume} {69}},\ \bibinfo {pages} {012306} (\bibinfo {year} {2004})}\BibitemShut {NoStop}%
\bibitem [{\citenamefont {Ball}\ \emph {et~al.}(2004)\citenamefont {Ball}, \citenamefont {Dragan},\ and\ \citenamefont {Banaszek}}]{ball2004exploiting}%
  \BibitemOpen
  \bibfield  {author} {\bibinfo {author} {\bibfnamefont {Jonathan}\ \bibnamefont {Ball}}, \bibinfo {author} {\bibfnamefont {Andrzej}\ \bibnamefont {Dragan}}, \ and\ \bibinfo {author} {\bibfnamefont {Konrad}\ \bibnamefont {Banaszek}},\ }\bibfield  {title} {\enquote {\bibinfo {title} {Exploiting entanglement in communication channels with correlated noise},}\ }\href {https://doi.org/10.1103/PhysRevA.69.042324} {\bibfield  {journal} {\bibinfo  {journal} {Physical Review A}\ }\textbf {\bibinfo {volume} {69}},\ \bibinfo {pages} {042324} (\bibinfo {year} {2004})}\BibitemShut {NoStop}%
\bibitem [{\citenamefont {Banaszek}\ \emph {et~al.}(2004)\citenamefont {Banaszek}, \citenamefont {Dragan}, \citenamefont {Wasilewski},\ and\ \citenamefont {Radzewicz}}]{banaszek2004experimental}%
  \BibitemOpen
  \bibfield  {author} {\bibinfo {author} {\bibfnamefont {Konrad}\ \bibnamefont {Banaszek}}, \bibinfo {author} {\bibfnamefont {Andrzej}\ \bibnamefont {Dragan}}, \bibinfo {author} {\bibfnamefont {Wojciech}\ \bibnamefont {Wasilewski}}, \ and\ \bibinfo {author} {\bibfnamefont {Czes{\l}aw}\ \bibnamefont {Radzewicz}},\ }\bibfield  {title} {\enquote {\bibinfo {title} {Experimental demonstration of entanglement-enhanced classical communication over a quantum channel with correlated noise},}\ }\href {https://doi.org/10.1103/PhysRevLett.92.257901} {\bibfield  {journal} {\bibinfo  {journal} {Physical Review Letters}\ }\textbf {\bibinfo {volume} {92}},\ \bibinfo {pages} {257901} (\bibinfo {year} {2004})}\BibitemShut {NoStop}%
\bibitem [{\citenamefont {Chiribella}\ and\ \citenamefont {Kristj{\'a}nsson}(2019)}]{chiribella2019quantum}%
  \BibitemOpen
  \bibfield  {author} {\bibinfo {author} {\bibfnamefont {Giulio}\ \bibnamefont {Chiribella}}\ and\ \bibinfo {author} {\bibfnamefont {Hl{\'e}r}\ \bibnamefont {Kristj{\'a}nsson}},\ }\bibfield  {title} {\enquote {\bibinfo {title} {Quantum shannon theory with superpositions of trajectories},}\ }\href {https://doi.org/10.1098/rspa.2018.0903} {\bibfield  {journal} {\bibinfo  {journal} {Proceedings of the Royal Society A}\ }\textbf {\bibinfo {volume} {475}},\ \bibinfo {pages} {20180903} (\bibinfo {year} {2019})}\BibitemShut {NoStop}%
\bibitem [{\citenamefont {Guha}\ \emph {et~al.}(2023)\citenamefont {Guha}, \citenamefont {Roy},\ and\ \citenamefont {Chiribella}}]{guha2023quantum}%
  \BibitemOpen
  \bibfield  {author} {\bibinfo {author} {\bibfnamefont {Tamal}\ \bibnamefont {Guha}}, \bibinfo {author} {\bibfnamefont {Saptarshi}\ \bibnamefont {Roy}}, \ and\ \bibinfo {author} {\bibfnamefont {Giulio}\ \bibnamefont {Chiribella}},\ }\bibfield  {title} {\enquote {\bibinfo {title} {Quantum networks boosted by entanglement with a control system},}\ }\href {https://doi.org/10.1103/PhysRevResearch.5.033214} {\bibfield  {journal} {\bibinfo  {journal} {Physical Review Research}\ }\textbf {\bibinfo {volume} {5}},\ \bibinfo {pages} {033214} (\bibinfo {year} {2023})}\BibitemShut {NoStop}%
\bibitem [{\citenamefont {Lai}\ \emph {et~al.}(2024)\citenamefont {Lai}, \citenamefont {Lin}, \citenamefont {Huang}, \citenamefont {Jan},\ and\ \citenamefont {Chen}}]{lai2024quick}%
  \BibitemOpen
  \bibfield  {author} {\bibinfo {author} {\bibfnamefont {Po-Rong}\ \bibnamefont {Lai}}, \bibinfo {author} {\bibfnamefont {Jhen-Dong}\ \bibnamefont {Lin}}, \bibinfo {author} {\bibfnamefont {Yi-Te}\ \bibnamefont {Huang}}, \bibinfo {author} {\bibfnamefont {Hsien-Chao}\ \bibnamefont {Jan}}, \ and\ \bibinfo {author} {\bibfnamefont {Yueh-Nan}\ \bibnamefont {Chen}},\ }\bibfield  {title} {\enquote {\bibinfo {title} {Quick charging of a quantum battery with superposed trajectories},}\ }\href {https://doi.org/10.1103/PhysRevResearch.6.023136} {\bibfield  {journal} {\bibinfo  {journal} {Physical Review Research}\ }\textbf {\bibinfo {volume} {6}},\ \bibinfo {pages} {023136} (\bibinfo {year} {2024})}\BibitemShut {NoStop}%
\bibitem [{\citenamefont {Saha}\ and\ \citenamefont {Sen}(2025)}]{saha2025interference}%
  \BibitemOpen
  \bibfield  {author} {\bibinfo {author} {\bibfnamefont {Sutapa}\ \bibnamefont {Saha}}\ and\ \bibinfo {author} {\bibfnamefont {Ujjwal}\ \bibnamefont {Sen}},\ }\bibfield  {title} {\enquote {\bibinfo {title} {Interference between lossy quantum evolutions activates information backflow},}\ }\href {https://doi.org/10.48550/arXiv.2507.22150} {\bibfield  {journal} {\bibinfo  {journal} {arXiv preprint arXiv:2507.22150}\ } (\bibinfo {year} {2025})}\BibitemShut {NoStop}%
\bibitem [{\citenamefont {Rubino}\ \emph {et~al.}(2021)\citenamefont {Rubino}, \citenamefont {Rozema}, \citenamefont {Ebler}, \citenamefont {Kristj{\'a}nsson}, \citenamefont {Salek}, \citenamefont {Allard~Gu{\'e}rin}, \citenamefont {Abbott}, \citenamefont {Branciard}, \citenamefont {Brukner}, \citenamefont {Chiribella} \emph {et~al.}}]{rubino2021experimental}%
  \BibitemOpen
  \bibfield  {author} {\bibinfo {author} {\bibfnamefont {Giulia}\ \bibnamefont {Rubino}}, \bibinfo {author} {\bibfnamefont {Lee~A}\ \bibnamefont {Rozema}}, \bibinfo {author} {\bibfnamefont {Daniel}\ \bibnamefont {Ebler}}, \bibinfo {author} {\bibfnamefont {Hl{\'e}r}\ \bibnamefont {Kristj{\'a}nsson}}, \bibinfo {author} {\bibfnamefont {Sina}\ \bibnamefont {Salek}}, \bibinfo {author} {\bibfnamefont {Philippe}\ \bibnamefont {Allard~Gu{\'e}rin}}, \bibinfo {author} {\bibfnamefont {Alastair~A}\ \bibnamefont {Abbott}}, \bibinfo {author} {\bibfnamefont {Cyril}\ \bibnamefont {Branciard}}, \bibinfo {author} {\bibfnamefont {{\v{C}}aslav}\ \bibnamefont {Brukner}}, \bibinfo {author} {\bibfnamefont {Giulio}\ \bibnamefont {Chiribella}},  \emph {et~al.},\ }\bibfield  {title} {\enquote {\bibinfo {title} {Experimental quantum communication enhancement by superposing trajectories},}\ }\href {https://doi.org/10.1103/PhysRevResearch.3.013093} {\bibfield  {journal} {\bibinfo  {journal} {Physical Review Research}\ }\textbf {\bibinfo
  {volume} {3}},\ \bibinfo {pages} {013093} (\bibinfo {year} {2021})}\BibitemShut {NoStop}%
\bibitem [{\citenamefont {Gregoratti}\ and\ \citenamefont {Werner}(2004)}]{gregoratti2004quantum}%
  \BibitemOpen
  \bibfield  {author} {\bibinfo {author} {\bibfnamefont {M}~\bibnamefont {Gregoratti}}\ and\ \bibinfo {author} {\bibfnamefont {RF}~\bibnamefont {Werner}},\ }\bibfield  {title} {\enquote {\bibinfo {title} {On quantum error-correction by classical feedback in discrete time},}\ }\href {https://doi.org/10.1063/1.1758320} {\bibfield  {journal} {\bibinfo  {journal} {Journal of Mathematical Physics}\ }\textbf {\bibinfo {volume} {45}},\ \bibinfo {pages} {2600--2612} (\bibinfo {year} {2004})}\BibitemShut {NoStop}%
\bibitem [{\citenamefont {Winter}(2005)}]{winter2005environment}%
  \BibitemOpen
  \bibfield  {author} {\bibinfo {author} {\bibfnamefont {Andreas}\ \bibnamefont {Winter}},\ }\bibfield  {title} {\enquote {\bibinfo {title} {On environment-assisted capacities of quantum channels},}\ }\href {https://doi.org/10.48550/arXiv.quant-ph/0507045} {\bibfield  {journal} {\bibinfo  {journal} {arXiv preprint quant-ph/0507045}\ } (\bibinfo {year} {2005})}\BibitemShut {NoStop}%
\bibitem [{\citenamefont {Gour}(2024)}]{gour2024resources}%
  \BibitemOpen
  \bibfield  {author} {\bibinfo {author} {\bibfnamefont {Gilad}\ \bibnamefont {Gour}},\ }\bibfield  {title} {\enquote {\bibinfo {title} {Resources of the quantum world},}\ }\href {https://doi.org/10.48550/arXiv.2402.05474} {\bibfield  {journal} {\bibinfo  {journal} {arXiv preprint arXiv:2402.05474}\ } (\bibinfo {year} {2024})}\BibitemShut {NoStop}%
\bibitem [{\citenamefont {Wootters}(1998)}]{wootters1998entanglement}%
  \BibitemOpen
  \bibfield  {author} {\bibinfo {author} {\bibfnamefont {William~K}\ \bibnamefont {Wootters}},\ }\bibfield  {title} {\enquote {\bibinfo {title} {Entanglement of formation of an arbitrary state of two qubits},}\ }\href {https://doi.org/10.1103/PhysRevLett.80.2245} {\bibfield  {journal} {\bibinfo  {journal} {Physical Review Letters}\ }\textbf {\bibinfo {volume} {80}},\ \bibinfo {pages} {2245} (\bibinfo {year} {1998})}\BibitemShut {NoStop}%
\bibitem [{\citenamefont {Horodecki}\ \emph {et~al.}(2000)\citenamefont {Horodecki}, \citenamefont {Horodecki},\ and\ \citenamefont {Horodecki}}]{horodecki2000limits}%
  \BibitemOpen
  \bibfield  {author} {\bibinfo {author} {\bibfnamefont {Micha{\l}}\ \bibnamefont {Horodecki}}, \bibinfo {author} {\bibfnamefont {Pawe{\l}}\ \bibnamefont {Horodecki}}, \ and\ \bibinfo {author} {\bibfnamefont {Ryszard}\ \bibnamefont {Horodecki}},\ }\bibfield  {title} {\enquote {\bibinfo {title} {Limits for entanglement measures},}\ }\href {https://doi.org/10.1103/PhysRevLett.84.2014} {\bibfield  {journal} {\bibinfo  {journal} {Physical Review Letters}\ }\textbf {\bibinfo {volume} {84}},\ \bibinfo {pages} {2014} (\bibinfo {year} {2000})}\BibitemShut {NoStop}%
\bibitem [{\citenamefont {Francica}\ \emph {et~al.}(2017)\citenamefont {Francica}, \citenamefont {Goold}, \citenamefont {Plastina},\ and\ \citenamefont {Paternostro}}]{Francica2017}%
  \BibitemOpen
  \bibfield  {author} {\bibinfo {author} {\bibfnamefont {Gianluca}\ \bibnamefont {Francica}}, \bibinfo {author} {\bibfnamefont {John}\ \bibnamefont {Goold}}, \bibinfo {author} {\bibfnamefont {Francesco}\ \bibnamefont {Plastina}}, \ and\ \bibinfo {author} {\bibfnamefont {Mauro}\ \bibnamefont {Paternostro}},\ }\bibfield  {title} {\enquote {\bibinfo {title} {Daemonic ergotropy: enhanced work extraction from quantum correlations},}\ }\href {\doibase 10.1038/s41534-017-0012-8} {\bibfield  {journal} {\bibinfo  {journal} {npj Quantum Information}\ }\textbf {\bibinfo {volume} {3}},\ \bibinfo {pages} {12} (\bibinfo {year} {2017})}\BibitemShut {NoStop}%
\bibitem [{\citenamefont {Biswas}\ \emph {et~al.}(2025)\citenamefont {Biswas}, \citenamefont {Datta},\ and\ \citenamefont {Garc{\'\i}a-Pintos}}]{biswas2025quantum}%
  \BibitemOpen
  \bibfield  {author} {\bibinfo {author} {\bibfnamefont {Tanmoy}\ \bibnamefont {Biswas}}, \bibinfo {author} {\bibfnamefont {Chandan}\ \bibnamefont {Datta}}, \ and\ \bibinfo {author} {\bibfnamefont {Luis~Pedro}\ \bibnamefont {Garc{\'\i}a-Pintos}},\ }\bibfield  {title} {\enquote {\bibinfo {title} {Quantum thermodynamic advantage in work extraction from steerable quantum correlations},}\ }\href {https://doi.org/10.1103/9qcc-7lq5} {\bibfield  {journal} {\bibinfo  {journal} {Physical Review Letters}\ }\textbf {\bibinfo {volume} {135}},\ \bibinfo {pages} {110402} (\bibinfo {year} {2025})}\BibitemShut {NoStop}%
\bibitem [{\citenamefont {Heinosaari}\ and\ \citenamefont {Ziman}(2011)}]{heinosaari2011mathematical}%
  \BibitemOpen
  \bibfield  {author} {\bibinfo {author} {\bibfnamefont {Teiko}\ \bibnamefont {Heinosaari}}\ and\ \bibinfo {author} {\bibfnamefont {M{\'a}rio}\ \bibnamefont {Ziman}},\ }\href {https://doi.org/10.1017/CBO9781139031103} {\emph {\bibinfo {title} {The mathematical language of quantum theory: from uncertainty to entanglement}}}\ (\bibinfo  {publisher} {Cambridge University Press},\ \bibinfo {year} {2011})\BibitemShut {NoStop}%
\bibitem [{\citenamefont {Lepp{\"a}j{\"a}rvi}\ and\ \citenamefont {Sedl{\'a}k}(2021)}]{leppajarvi2021postprocessing}%
  \BibitemOpen
  \bibfield  {author} {\bibinfo {author} {\bibfnamefont {Leevi}\ \bibnamefont {Lepp{\"a}j{\"a}rvi}}\ and\ \bibinfo {author} {\bibfnamefont {Michal}\ \bibnamefont {Sedl{\'a}k}},\ }\bibfield  {title} {\enquote {\bibinfo {title} {Postprocessing of quantum instruments},}\ }\href {https://doi.org/10.1103/PhysRevA.103.022615} {\bibfield  {journal} {\bibinfo  {journal} {Physical Review A}\ }\textbf {\bibinfo {volume} {103}},\ \bibinfo {pages} {022615} (\bibinfo {year} {2021})}\BibitemShut {NoStop}%
\bibitem [{\citenamefont {Ghai}\ and\ \citenamefont {Mitra}(2025)}]{ghai2025instrument}%
  \BibitemOpen
  \bibfield  {author} {\bibinfo {author} {\bibfnamefont {Jatin}\ \bibnamefont {Ghai}}\ and\ \bibinfo {author} {\bibfnamefont {Arindam}\ \bibnamefont {Mitra}},\ }\bibfield  {title} {\enquote {\bibinfo {title} {Instrument-based quantum resources: quantification, hierarchies and towards constructing resource theories},}\ }\href {https://doi.org/10.48550/arXiv.2508.09134} {\bibfield  {journal} {\bibinfo  {journal} {arXiv preprint arXiv:2508.09134}\ } (\bibinfo {year} {2025})}\BibitemShut {NoStop}%
\bibitem [{\citenamefont {Sau}\ \emph {et~al.}(2026)\citenamefont {Sau}, \citenamefont {Jen{\v{c}}ov{\'a}},\ and\ \citenamefont {Guha}}]{sau2026demultiplexing}%
  \BibitemOpen
  \bibfield  {author} {\bibinfo {author} {\bibfnamefont {Soham}\ \bibnamefont {Sau}}, \bibinfo {author} {\bibfnamefont {Anna}\ \bibnamefont {Jen{\v{c}}ov{\'a}}}, \ and\ \bibinfo {author} {\bibfnamefont {Tamal}\ \bibnamefont {Guha}},\ }\bibfield  {title} {\enquote {\bibinfo {title} {Demultiplexing generalized information via quantum transmission lines},}\ }\href {https://doi.org/10.48550/arXiv.2606.17894} {\bibfield  {journal} {\bibinfo  {journal} {arXiv preprint arXiv:2606.17894}\ } (\bibinfo {year} {2026})}\BibitemShut {NoStop}%
\bibitem [{\citenamefont {Gisin}(1989)}]{gisin1989stochastic}%
  \BibitemOpen
  \bibfield  {author} {\bibinfo {author} {\bibfnamefont {Nicolas}\ \bibnamefont {Gisin}},\ }\bibfield  {title} {\enquote {\bibinfo {title} {Stochastic quantum dynamics and relativity},}\ }\href {https://doi.org/10.5169/SEALS-116034} {\bibfield  {journal} {\bibinfo  {journal} {Helv. Phys. Acta}\ }\textbf {\bibinfo {volume} {62}},\ \bibinfo {pages} {363--371} (\bibinfo {year} {1989})}\BibitemShut {NoStop}%
\bibitem [{\citenamefont {Hughston}\ \emph {et~al.}(1993)\citenamefont {Hughston}, \citenamefont {Jozsa},\ and\ \citenamefont {Wootters}}]{hughston1993complete}%
  \BibitemOpen
  \bibfield  {author} {\bibinfo {author} {\bibfnamefont {Lane~P}\ \bibnamefont {Hughston}}, \bibinfo {author} {\bibfnamefont {Richard}\ \bibnamefont {Jozsa}}, \ and\ \bibinfo {author} {\bibfnamefont {William~K}\ \bibnamefont {Wootters}},\ }\bibfield  {title} {\enquote {\bibinfo {title} {A complete classification of quantum ensembles having a given density matrix},}\ }\href {https://doi.org/10.1016/0375-9601(93)90880-9} {\bibfield  {journal} {\bibinfo  {journal} {Physics Letters A}\ }\textbf {\bibinfo {volume} {183}},\ \bibinfo {pages} {14--18} (\bibinfo {year} {1993})}\BibitemShut {NoStop}%
\bibitem [{\citenamefont {Piani}\ and\ \citenamefont {Watrous}(2015)}]{piani2015necessary}%
  \BibitemOpen
  \bibfield  {author} {\bibinfo {author} {\bibfnamefont {Marco}\ \bibnamefont {Piani}}\ and\ \bibinfo {author} {\bibfnamefont {John}\ \bibnamefont {Watrous}},\ }\bibfield  {title} {\enquote {\bibinfo {title} {Necessary and sufficient quantum information characterization of {Einstein-Podolsky-Rosen} steering},}\ }\href {https://doi.org/10.1103/PhysRevLett.114.060404} {\bibfield  {journal} {\bibinfo  {journal} {Physical Review Letters}\ }\textbf {\bibinfo {volume} {114}},\ \bibinfo {pages} {060404} (\bibinfo {year} {2015})}\BibitemShut {NoStop}%
\bibitem [{\citenamefont {Coffman}\ \emph {et~al.}(2000)\citenamefont {Coffman}, \citenamefont {Kundu},\ and\ \citenamefont {Wootters}}]{coffman2000distributed}%
  \BibitemOpen
  \bibfield  {author} {\bibinfo {author} {\bibfnamefont {Valerie}\ \bibnamefont {Coffman}}, \bibinfo {author} {\bibfnamefont {Joydip}\ \bibnamefont {Kundu}}, \ and\ \bibinfo {author} {\bibfnamefont {William~K}\ \bibnamefont {Wootters}},\ }\bibfield  {title} {\enquote {\bibinfo {title} {Distributed entanglement},}\ }\href {https://doi.org/10.1103/PhysRevA.61.052306} {\bibfield  {journal} {\bibinfo  {journal} {Physical Review A}\ }\textbf {\bibinfo {volume} {61}},\ \bibinfo {pages} {052306} (\bibinfo {year} {2000})}\BibitemShut {NoStop}%
\end{thebibliography}%
\section*{Appendix}
\subsection{Proof of Proposition \ref{p1}}\label{pp1}
Let us recall the expression of \(W_{weak}(\rho_s,\Lambda_{\beta})\) in  Definition \ref{d1},
\begin{align}\label{a1}
    \nonumber W_{weak}(\rho_s,\Lambda_{\beta})&=\max_{\{\Pi_b^k\}}\sum_kp_kF_{\beta}(\sigma_{s|k})\\\nonumber&=\max_{\{\Pi_b^k\}}\sum_kp_k[\Tr(H_s\sigma_{s|k})-\frac 1{\beta}S(\sigma_{s|k})]\\\nonumber&=\Tr(H_s\sum_kp_k\sigma_{s|k})-\frac 1{\beta}\min_{\{\Pi_b^k\}}\sum_kp_kS(\sigma_{s|k})\\&=E(\sigma_s)-\frac 1{\beta}\min_{\{\Pi_b^k\}}\sum_kp_kS(\sigma_{s|k}),
\end{align}
where all the symbols have usual meaning as in the main text and we have used the fact that \(\sum_k p_k\sigma_{s|k}=\sigma_s=\Lambda_{\beta}(\rho_s)\) for any possible choice of the measurement \(\{\Pi_b^k\}\). 

Now, using the concavity of the von Neumann entropy, we have 
\[\sum_k p_kS(\sigma_{s|k})\leq S(\sum_k p_k\sigma_{s|k})=S(\sigma_s).\]
This, together with Eq. (\ref{a1}) implies 
\begin{align}\label{a2}
W_{weak}(\rho_s,\Lambda_{\beta})\geq E(\sigma_s)-\frac 1{\beta}S(\sigma_s)=F_{\beta}(\sigma_s).
\end{align}
Let us now consider the charge retrieval under strong assistance with a specific choice of measurements for both the bath and the reference systems. In particular, we choose \(\{\Pi_k^{b}\}^*\) as the optimal measurement to achieve \(W_{weak}(\rho_s,\Lambda_{\beta})\) and the trivial measurement \(\{T_R^0=\mathbb{I}_R,T_R^1=\Theta_R\}\), \(\Theta_R\) being the null operator, for the reference system. It is then easy to argue that
\begin{align}\label{a3}
W_{strong}(\rho_s,\Lambda_{\beta},\{(\Pi_b^k)^*\otimes T_R^l\})\leq W_{strong}(\rho_s,\Lambda_{\beta}),
\end{align}
since the latter involves optimization over all possible product measurements allowed to perform on the bath and the reference systems.

Now, for \(W_{strong}(\rho_s,\Lambda_{\beta},\{(\Pi_b^k)^*\otimes T_R^l\})\), we have
\begin{align*}
p_{k,l}&=\Tr[(\mathbb{I}_s\otimes(\Pi_b^k)^*\otimes T_R^l)\sigma_{sbR}]\\
&=\Tr[(\mathbb{I}_s\otimes(\Pi_b^k)^*)\Tr_R(\sigma_{sbR})]\times\delta_{l,0}\\&=p_k\times\delta_{l,0}\\\text{and }\sigma_{s|k,0}&=\frac 1{p_{k,0}}\Tr_{bR}[(\mathbb{I}_s\otimes(\Pi_b^k)^*\otimes\mathbb{I}_R)\sigma_{sbR}]\\&=\frac 1{p_k}\Tr_b[(\mathbb{I}_s\otimes(\Pi_b^k)^*)\Tr_R(\sigma_{sbR})].
\end{align*}
Therefore, we can identify
\[W_{strong}(\rho_s,\Lambda_{\beta},\{(\Pi_b^k)^*\otimes T_R^l\})=W_{weak}(\rho_s,\Lambda_{\beta}).\]
Hence using Eq. (\ref{a3}), we can conclude
\begin{align}\label{a4}
    W_{weak}(\rho_s,\Lambda_{\beta})\leq W_{strong}(\rho_s,\Lambda_{\beta}).
\end{align}
Finally, using the analysis similar to that of Eq. (\ref{a1}) we have 
\begin{align}\label{a4.5}
    W_{strong}(\rho_s,\Lambda_{\beta})=E(\sigma_s)-\frac 1{\beta}\min_{\{\Pi_b^k\otimes\Pi_R^l\}}\sum_{k,l}p_{k,l}S(\sigma_{s|k,l}).
\end{align}
This, together with the positivity of the von Neumann entropy implies
\begin{align}\label{a5}
    W_{strong}(\rho_s,\Lambda_{\beta})\leq E(\sigma_s). 
\end{align}
Therefore, from Eqs. (\ref{a2}),  (\ref{a4}) and (\ref{a5}) together, we conclude the proof of the proposition.
\subsection{Proof of Theorem \ref{t1}}\label{pt1}
Let us begin with the following two lemmas, which will be instrumental for some of the results.
\begin{lemma}\label{l1}
For a bipartite state \(\rho_{AB}\), with a given decomposition \(\rho_{AB}=\sum_i q_i\eta^i_{AB}\), consider the following quantity
\[S_{\{q_i,\eta_i\}}(A)=\sum_i q_iS(\eta^i_A),\text{ where, }\eta_A^i=\Tr_B[\eta^i_{AB}].\]
The value of \(S_{\{q_i,\eta_i\}}(A)\) reaches to its minimum, only for a possible pure decomposition of \(\rho_{AB}\). 
\end{lemma}
\begin{proof}
    Suppose the optimal decomposition of \(\rho_{AB}=\sum_ip_i\sigma_{AB}^i\), such that
    \begin{align*}
        S_{opt}(A)=S_{\{p_i,\sigma_i\}}(A)=\sum_i p_iS(\sigma_A^i)
    \end{align*}
    and there exists atleast one \(\sigma_{AB}^{i^{\star}}\) which is mixed in nature. That is,
    \[\sigma_{AB}^{i^{\star}}=\sum_{k_i}t_{k_i}\ketbra{\psi_{k_i}}{\psi_{k_i}}_{AB}.\] 
    But then, by denoting \(\sigma_A^{k_i}=\Tr_B(\ketbra{\psi_{k_i}}{\psi_{k_i}})\) and using the concavity of von Neumann entropy,
    \begin{align*}
      S_{\{p_it_{k_i},\ketbra{\psi_{k_i}}{\psi_{k_i}}\}}(A)&=\sum_{i\neq i^*}p_iS(\sigma^i_A)+p_{i^{\star}}\sum_{k_i}t_{k_i}S(\sigma_A^{k_i})\\&\leq \sum_{i\neq i^*}p_iS(\sigma^i_A)+p_{i^{\star}}S(\sum_{k_i}t_{k_i}\sigma_A^{k_i})\\&=\sum_{i\neq i^*}p_iS(\sigma^i_A)+p_{i^{\star}}S(\sigma_A^{i^{\star}})\\&=S_{opt}(A)  
    \end{align*}
    This leads to a contradiction and hence completes the proof.
\end{proof}
\begin{lemma}\label{l2}
     Consider a bipartite state \(\rho_{AB}\). Then the minimal average entropy of the marginal  \(A\), optimized over all possible POVM for \(B\) equals to \(E_f(\rho_{AC})\). Here, \(E_f(\rho_{AC})\) is the entanglement of formation of the state \(\rho_{AC}=\Tr_B(\ketbra{\psi}{\psi}_{ABC})\) and \(\ket{\psi}_{ABC}\) is the purification of \(\rho_{AB}\), that is, \(\rho_{AB}=\Tr_C(\ketbra{\psi}{\psi}_{ABC})\).
\end{lemma}
\begin{proof}
    The minimal average entropy at \(A\), conditioned over the measurement outcomes, optimized over all possible POVMs at \(B\), can be written as
    \[S_{opt}(A|B)=\min_{\{\Pi_k\}_k}\sum_k p_k S(\rho_{A|k}),\]
    where \(p_k\) is the probability that the POVM element \(\Pi_k\) clicks at \(B\) and \(\rho_{A|k}\) is the updated state at \(A\) in that particular instant. 
    
    Also note that every possible measurement \(\{\Pi_k\}_k\) at \(B\), produces different decompositions of the state \(\rho_{AC}\) \cite{gisin1989stochastic, hughston1993complete}. Thanks to Lemma \ref{l1}, we can now restrict ourselves in the particular set of POVMs \(\{M_k\}_k\)---only with  rank one elements \cite{piani2015necessary}, producing all possible pure decompositions of \(\rho_{AC}=\sum_k p_k \ketbra{\psi_k}{\psi_k}_{AC}\). So,
    \begin{align}\label{a6}
    S_{opt}(A|B)&=\min_{\{M_k\}_k}\sum_kp_k S(\rho_{A|k}),\\\nonumber
    \text{where, }p_k&=\Tr[(\mathbb{I}_A\otimes {M_k})\rho_{AB}]\\\nonumber&=\Tr[(\mathbb{I}_A\otimes {M_k}\otimes\mathbb{I}_C)\ketbra{\psi}{\psi}_{ABC}],\\\nonumber\text{and }\rho_{A|k}&=\Tr_B[(\mathbb{I}_A\otimes {M_k}_)\rho_{AB}]\\\nonumber&=\Tr_C(\ketbra{\psi_k}{\psi_k}_{AC}).
    \end{align}
    Now, for any bipartite pure state \(\ketbra{\psi_k}{\psi_k}_{AC}\), the entanglement of formation \cite{wootters1998entanglement} (in fact, all possible entanglement measures \cite{horodecki2000limits})
    \[E_f(\ketbra{\psi_k}{\psi_k}_{AC})=S(\rho_{A|k})=S(\rho_{C|k}).\]
    Therefore, using this expression in Eq. (\ref{a6}), we have
    \begin{align*}
    S_{opt}(A|B)&=\min_{\{p_k,\psi_k\}} \sum_k p_k E_f(\ketbra{\psi_k}{\psi_k}_{AC})\\&=E_f(\rho_{AC})
    \end{align*}
   Hence this completes the proof.
   \end{proof}
   Coming back to the proof of the main Theorem, it is first important to recall the expression of the weak retrieval of charge in a quantum battery, as in Eq. (\ref{a1}), and then applying Lemma \ref{l2} we  have
   \begin{align}\label{a7}
       \nonumber W_{weak}(\rho_s,\Lambda_{\beta})&=E(\sigma_s)-\frac 1{\beta} \min_{\{\Pi_b^k \}}\sum_kp_kS(\sigma_{s|k})\\&=E(\sigma_s)-\frac 1{\beta} E_f(\sigma_{sRR'})
   \end{align}
   where, \(\sigma_{sRR'}=\Tr_b(\ketbra{\psi}{\psi}_{sbRR'})\) and \(\ket{\psi}_{sbRR'}\) is the purification of the state 
   \[\sigma_{sbR}=(U_{sb}\otimes\mathbb{I}_R)(\rho_s\otimes\ketbra{\phi^+_{\beta}}{\phi^+_{\beta}}_{bR})(U_{sb}^{\dagger}\otimes\mathbb{I}_R)\]
Also note that the \(E_f(\sigma_{sRR'})\) implies the entanglement of formation of the state \(\sigma_{sRR'}\) in the \(s|RR'\) bipartition. It is now trivial to argue that
\begin{align}\label{a8}
E_f(\sigma_{sR})\leq E_f(\sigma_{sRR'}),
\end{align}
where, \(\sigma_{sR}=\Tr_{R'}[\sigma_{sRR'}]\), i.e., just by tracing out the subsystem \(R'\) and since, any entanglement measure (including \(E_f\)) non-increasing under LOCC. Therefore, by invoking Eq. (\ref{a8}) in Eq. (\ref{a7}), we can conclude
\[W_{weak}(\rho_s,\Lambda_{\beta})\leq E(\sigma_s)-\frac 1{\beta} E_f(\sigma_{sR}).\]
Moreover, when \(\rho_s\) is a pure quantum battery, then \(\sigma_{sbR}\) is also pure. Hence, in that case \(E_f(\sigma_{sRR'})=E_f(\sigma_{sR})\). This trivially implies that the above inequality saturates for \(\rho_s\) being pure and therefore,  completes the proof.
\subsection{Proof of Theorem \ref{t2}}\label{pt2}
Let us first recall the expression of charge retrieval in a qubit battery, under strong assistance, as in Eq. (\ref{a4.5}), and using Lemma \ref{l1}, we can again restrict the local POVMs at \(b\) and \(R\), with rank one elements \(\{M^k_b\otimes M^l_R\}_{k,l}\) only, that is
\begin{align}\label{a8.5}
W_{strong}(\rho_s,\Lambda_{\beta})=E(\sigma_s)-\frac 1{\beta}\min_{\{M_b^k\otimes M_R^l\}_{k,l}}\sum_{k,l}p_{k,l}S(\sigma_{s|k,l}).\end{align}
Also note that
\begin{align}\label{a9}
    \nonumber\min_{\{M_b^k\otimes M_R^l\}_{k,l}}\sum_{k,l}p_{k,l}S(\sigma_{s|k,l})&\geq \min_{\{M_{bR}^m\}_m}\sum_m p_m S(\sigma_{s|m})\\&=E_f(\sigma_{sR'})
\end{align}
where, \(\{M^m_{bR}\}_m\) is the one-rank POVM acting jointly on the bath (\(b\)) and the reference (\(R\)) subsystems and \(\sigma_{sR'}=\Tr_{bR}(\ketbra{\psi}{\psi}_{sbRR'})\). Then Eq. (\ref{a8.5}) can be written in terms of the following inequality
\begin{align}\label{a9.25}
   W_{strong}(\rho_s,\Lambda_{\beta})\leq E(\sigma_s)-\frac 1{\beta}E_f(\sigma_{sR'}). 
\end{align}
As mentioned earlier, whenever the input quantum battery \(\rho_s\) is pure, the state \(\sigma_{sbR}\) is also pure. Therefore, \(E_f(\sigma_{sR'})=0\), for the pure quantum state \(\rho_s\). 

Nevertheless, for the pure tripartite state \(\ket{\psi}_{sbR}\) we can choose any arbitrary product basis \(\{\{\ket{\psi_i}_b\otimes\ket{\phi_j}_R\}_{i=0}^{d_b-1}\}_{j=0}^{d_R-1}\), such that
\begin{align}\label{a9.5}
    \nonumber\ket{\phi_{\psi}^{\beta}}_{sbR}=\sum_{i=0}^{d_b-1}\sum_{j=0}^{d_R-1}c_{ij}\ket{\xi_{i,j}}_s&\otimes\ket{\psi_i}_b\otimes\ket{\phi_j}_R,\\\text{where, }\Tr_{bR}[\ketbra{\phi_{\psi}^{\beta}}{\phi_{\psi}^{\beta}}_{sbR}]=& \sigma_s=\Lambda_{\beta}(\ketbra{\psi}{\psi}_s)
\end{align}
Evidently, if the bath and the reference qudits are measured in the \(\{\ketbra{\psi_i}{\psi_i}\}_{i=0}^{d_b-1}\) and \(\{\ketbra{\phi_j}{\phi_j}\}_{j=0}^{d_R-1}\) bases respectively, then the post measurement states in the QB are pure---possessing no entropy. Therefore, using Eq. (\ref{a8.5})
\[W_{strong}(\rho_s,\Lambda_{\beta})=E(\sigma_s),\]
Note that whenever \(\rho_s\) is pure, then in fact, any projective measurement will work for the maximal amount of the charge retrieval under strong assistance. This completes the proof.
\subsection{Proof of Theorem \ref{t3}}\label{pt3}
\subsubsection{Proof of part (i)}
  We begin with the following two Lemmas. While Lemma \ref{l3} is simple, we will find a nontrivial implication of it in Lemma \ref{l4}, in the  question of identifying the Stinespring representation of a thermal map \(\Lambda_{\beta}(\rho_s)=\tau_{\beta}^s,~\forall \rho_s\), uniquely.
  \begin{lemma}\label{l3}
      Consider any N-partite pure quantum state \(\ket{\psi}_{A_1\cdots A_N}\), with the iso-spectral marginals for \(A_1,\cdots,A_k\), where \(k<N\) and 
      \(\rho_{A_j}=\sum_i p_i \ketbra{\psi^{(j)}_i}{\psi^{(j)}_i}_{A_j} \). Then the state admits the following canonical form (up to local unitary):
      \begin{align}\label{a10}
      \ket{\psi}_{A_1\cdots A_N}=\sum_i\sqrt{p_i}\bigotimes_{j=1}^k\ket{\psi_i^{(j)}}_{A_j}\otimes\ket{\phi_i}_{A_{k+1}\cdots A_N},
      \end{align}
      where, \(\forall j~,\langle\psi_i^{(j)}|\psi_k^{(j)}\rangle=\delta_{i,k}\) and \(\langle\phi_i|\phi_k\rangle=\delta_{i,k}\) whenever \(\ket{\phi_i}\nsim\ket{\phi_k}\).
  \end{lemma}
  \begin{proof}
      Consider the \(A_1\) vs. rest bipartition of the state \(\ket{\psi}_{A_1\cdots A_N}\). Since \(\rho_{A_1}=\sum_ip_i\ketbra{\psi_i^{(1)}}{\psi_i^{(1)}}_{A_1}\), we can write the pure state in its Schimdt form:
      \[\ket{\psi}_{A_1\cdots A_N}=\sum_i\sqrt{p_i} \ket{\psi_i^{(1)}}_{A_1}\otimes\ket{\xi_i}_{A_2\cdots A_N}.\]
      where, \(\langle\xi_i|\xi_k\rangle=\delta_{i,k},\) whenever \(\ket{\xi_i}\nsim\ket{\xi_k}\).
      This further implies that 
      \[\rho_{A_2\cdots A_N}=\sum_i p_i\ketbra{\xi_i}{\xi_i}_{A_2\cdots A_N}.\]
      But we have, %
      \begin{align*} 
      \rho_{A_2}&=\sum_ip_i\ketbra{\psi_i^{(2)}}{\psi_i^{(2)}}_{A_2}\\&=\Tr_{A_3\cdots A_N}[\rho_{A_2\cdots A_N}]\\&=\sum_i p_i\Tr_{A_3\cdots A_N}[\ketbra{\xi_i}{\xi_i}_{A_2\cdots A_N}
      ].\end{align*}
      Therefore,
      \[\forall i,~\ket{\xi_i}_{A_2\cdots A_N}=\ket{\psi_i^{(2)}}_{A_2}\otimes \ket{\eta_i}_{A_3\cdots A_N},\]
      and hence,
      \[\ket{\psi}_{A_1\cdots A_N}=\sum_i\sqrt{p_i} \ket{\psi_i^{(1)}}_{A_1}\otimes\ket{\psi_i^{(2)}}_{A_2}\otimes\ket{\eta_i}_{A_3\cdots A_N}.\]
      where, we again have \(\langle\eta_i|\eta_k\rangle=\delta_{i,k},\) whenever \(\ket{\eta_i}\nsim\ket{\eta_k}\). We can now trivially extend similar representation for each of the parties \(\{A_3,\cdots,A_k\}\), which leads to the final form as in Eq. (\ref{a10}).

      Also note that, this canonical form is unique up to local unitary for each of the parties \(\{A_1,\cdots,A_k\}\), where each \(U_j,~j\in\{1,\cdots,k\}\) is diagonalized in the corresponding \(\{\ket{\psi_i^{(j)}}\}_i\) basis.
  \end{proof}

  \begin{lemma}\label{l4}
      For completely non-degenerate \(H_s=H_b\), the energy-conserving joint unitary for every thermal map is uniquely the SWAP operation between the QB and the bath qudit (upto local energy preserving unitary).
  \end{lemma}
  \begin{proof}
      A Thermal map,being a thermal operation, possesses in general a Stinespring representation of the form
      \begin{align}\label{a11}
          \nonumber\Lambda_{\beta}(\rho_s)=\Tr_b[U_{sb}&(\rho_s\otimes\tau_{\beta}^b)U_{sb}^{\dagger}]=\Tr_b[\sigma_{sb}^{\rho}]=\tau_{\beta}^s,\\\text{where, }[U_{sb},H_{tot}]&=0,~H_{tot}=H_s\otimes\mathbb{I}_b+\mathbb{I}_s\otimes H_b.
      \end{align}
      Now consider \(\rho_s=\ketbra{\psi}{\psi}_s\) a pure quantum battery and using the purification of the thermal bath \(\tau_{\beta}\) (as in Eq. (\ref{a3})), the action of \(U_{sb}\) produces a tripartite pure state \(\ket{\phi_{\psi}^{\beta}}_{sbR}\).

      The above thermal map definition for the QB, along with the non-signaling constraint implies
      \begin{align*}
          \sigma_s^{(\psi,\beta)}&=\Tr_{bR}[\ketbra{\phi_{\psi}^{\beta}}{\phi_{\psi}^{\beta}}_{sbR}]=\tau_{\beta}^s\\\text{and } \sigma_R^{\beta}&=\Tr_{sb}[\ketbra{\phi_{\psi}^{\beta}}{\phi_{\psi}^{\beta}}_{sbR}]=\tau_{\beta}^R.
      \end{align*}
      The condition \(H_s=H_b=\sum_kE_k\ketbra{E_k}{E_k}\) then implies, \[\tau_{\beta}^s=\tau_{\beta}^b=\tau_{\beta}^R=\frac1Z\sum_k\exp(-\beta E_k)\ketbra{E_k}{E_k}.\] Now, using Lemma \ref{l3}, we have the following canonical form 
      \begin{align}\label{a12}
         \ket{\phi_{\psi}^{\beta}}_{sbR}= \frac{1}{\sqrt{Z}}\sum_k\exp(-\frac{\beta E_k}{2})\ket{E_k}_s\otimes\ket{\xi_k^{\psi}}_b\otimes\ket{E_k}_R,
      \end{align}
   where \(\langle\xi_k^{\psi}|\xi_l^{\psi}\rangle=\delta_{k,l}\) whenever \(\ket{\xi_k}\nsim\ket{\xi_l}\) for every pure quantum battery \(\ket{\psi}_s\). 
   
   Now, consider \(\ket{\psi}_s=\ket{E_m}_s\), an arbitrary energy eigen state. Then using Eq. (\ref{a12}), we can conclude that the action of energy-preserving \(U_{sb}\) produces
   \begin{align}\label{a13}
U_{sb}\ket{E_m}_s\otimes\ket{E_k}_b=\ket{E_k}_s\otimes\ket{\xi_k^{E_m}}_b,~\forall \ket{E_k}_b,
   \end{align}
   where, \(\forall k,~ E(\ketbra{\xi_k^{E_m}}{\xi_k^{E_m}})=E_m\). Notably, Eq. (\ref{a12}) also implies that either, \(\{\ket{\xi_k^{E_m}}\}_k\) forms an orthonormal sub-basis or, \(\forall k, \ket{\xi_k^{E_m}}=e^{i\theta_k} \ket{E_m},\) up to energy preserving unitary. But the complete non-degeneracy of the Hamiltonian admits that the first one is impossible. Therefore, Eq. (\ref{a13}) must admit a particular form:
\[U_{sb}\ket{E_m}_s\otimes\ket{E_k}_b=e^{i\theta_k}\ket{E_k}_s\otimes\ket{E_m}_b,\forall\ket{E_k}_b,~\ket{E_m}_s.\]
   Using this in Eq. (\ref{a12}), we get
   \[\ket{\phi_{\psi}^{\beta}}_{sbR}= (\frac{1}{\sqrt{Z}}\sum_k\exp(-\frac{\beta E_k}{2}+i\theta_k)\ket{E_k}_s\otimes\ket{E_k}_R)\otimes\ket{\psi}_b.\]
   This proves the claim. 
   \end{proof}
  With Lemma \ref{l4}, it is now trivial to argue the first part of Theorem \ref{t3}. 
Note that after the action of the thermal map the initial state goes to
  \[\rho_s\otimes\ketbra{\phi^+_{\beta}}{\phi^+_{\beta}}_{bR} \mapsto \rho_b\otimes\ketbra{\phi^+_{\beta}}{\phi^+_{\beta}}_{sR}.\]
  So any measurement on the bath qudit reveals no information about the post measurement state of the QB. Therefore,
  \begin{align*}  
 W_{weak}(\rho_s,\Lambda_{\beta})&\leq E(\tau_{\beta}^s)-\frac 1{\beta} E_f(\ketbra{\phi^+_{\beta}}{\phi^+_{\beta}})\\&=E(\tau_{\beta}^s)-\frac 1{\beta}S(\tau_{\beta}^s)\\&=F_{\beta}(\tau_{\beta})=0 \text{ (by rescaling) } 
 \end{align*}
 Since the restored amount of charge for our analysis is conventionally non-negative, we have \(W_{weak}(\rho_s,\Lambda_{\beta})=0\) for every thermal map \(\Lambda_{\beta}\).
 
 On the other hand, for every pure input battery , the final state between QB and the reference \(R\) is pure entangled, which implies \(E_f(\sigma_{sR'})=0\). From Eq. (\ref{a8.5}), we can now immediately conclude
 \[W_{strong}(\rho_s,\Lambda_{\beta})=E(\tau_{\beta}).\]
 In fact, the optimal measurement to restore this amount of charge is independent of the rank-one POVM performed on the reference system \(R\).
 \subsubsection{Proof of part (ii)}
 When the thermal environment is of absolute zero temperature, that is, \(\beta\to \infty\), then using Eq. (\ref{e3}) the purified environment is simply \(\ket{\phi^+_{\infty}}_{bR}=\ket{0}_b\otimes\ket{0}_R\), by rescaling \(E_0=0\). Note that, the complete non-degeneracy of \(H_s=H_b\) indeed confirms that the purified environment must be of the product form. 
 
 Then for any unitary \(U_{sb}\) on the QB-bath qudit joint state, the reference system \(R\) remains completely uncorrelated due to the data processing inequality.

 Therefore, it is now trivial to argue that \(W_{weak}(\rho_s,\Lambda_{\infty})=W_{strong}(\rho_s,\Lambda_{\infty})\). In fact, from Eq. (\ref{a7}) one can estimate the amount of retrieved charge 
 \begin{align*}
 W_{weak}(\rho_s,\Lambda_{\infty})=W_{strong}(\rho_s,\Lambda_{\infty})=E(\sigma_s)-\frac1{\beta}E_f(\sigma_{s|R'}),
 \end{align*} 
 where \(\sigma_s=\Lambda_{\infty}(\rho_s)\). Hence,
 \[\Delta(\rho_s,\Lambda_{\infty})=0.\] 
 Moreover, when \(\rho_s\) is pure, we have \(E_f(\sigma_{s|R'})=0\) and hence the retrieved charge saturates the last three inequalities of Proposition \ref{p1}, independent of the rank one POVM performed on the bath qudit.

\subsection{Qubit Battery: Form of Thermal Operations}\label{q2b}
We will now consider a simple example where both the QB and the interacting thermal environment can be modeled as two-level quantum systems and governed by the same Hamiltonian. Without loss of  generality, here we assume \(H_s=H_b=\ketbra{1}{1}\). Also note that similar to Eq. (\ref{e3}), the bipartite purification of a thermal qubit \(\tau_{\beta}^b\) can be written as
\begin{equation}\label{e9}
    \ket{\phi^+_{\beta}}_{bR}=\frac 1{\sqrt{Z}}(\ket{0}_{b}\otimes\ket{0}_{R}+e^{-\frac{\beta}2}\ket{1}_b\otimes\ket{1}_{R}),
\end{equation}
where, \(Z=1+e^{-\beta}\) is the partition function. This, together with the Hamiltonian description, lead to the following Proposition.
\begin{proposition}\label{p2}
Any qubit thermal operation \(\Lambda_{\beta}\) admits the following isometric extension:
\begin{align}\label{e10}
    \mathcal{V}_{\Lambda_{\beta}}:&\begin{cases}\ket{0}_s\to\frac 1{\sqrt{Z}}(\ket{0_s0_b0_R}+e^{-\frac{\beta}2}\ket{\psi^+_{\alpha}}_{sb}\ket{1}_{R})=:\ket{\phi_0^{\beta}}\\
    \ket{1}_s\to\frac 1{\sqrt{Z}}(\ket{\psi^-_{\alpha}}_{sb}\ket{0}_{R}+e^{-\frac{\beta}2}\ket{1_s1_b1_R})=:\ket{\phi_1^{\beta}}
    \end{cases}
\\\nonumber\text{where, }&\ket{\psi^+_{\alpha}}=\alpha\ket{01}+\gamma\ket{10},\text{ and}
\\\nonumber&\ket{\psi^-_{\alpha}}=e^{i\phi}(\gamma^*\ket{01}-\alpha^*\ket{10})\\\nonumber&\text{ with }|\alpha|^2+|\gamma|^2=1\text{ and }\phi\in[0,2\pi).
\end{align}
\end{proposition}
\begin{proof}
The proof is intuitive. 

Using Eq. (\ref{e4}), we can decompose the battery-bath-reference joint unitary
\(U_{sbR}=U_{sb}\otimes\mathbb{I}_R\),
where, \(U_{sb}\) must preserve the total energy of system-bath joint state. 

With the assumption \(H_s=H_b=\ketbra{1}{1}\), we can immediately identify \(\ket{0}_s\otimes\ket{0}_b\) and \(\ket{1}_s\otimes\ket{1}_b\), respectively as the minimal and maximal energetic battery-bath joint states, which are also non-degenrate in nature. This implies,
\begin{align*}
    U_{sb}\ket{k}_s\otimes\ket{k}_b=\ket{k}_s\otimes\ket{k}_b,\quad \text{for }k\in\{0,1\}
\end{align*}
Also note that the subspace spanned by \(\{\ket{0}_s\otimes\ket{1}_b,\ket{1}_s\otimes\ket{0}_b\}\) is completely degenerate with respect to the total Hamiltonian \(H_{sb}=H_s\otimes \mathbb{I}+\mathbb{I}\otimes H_b\). Moreover, this is the only subspace orthogonal to both \(\{\ket{0}_s\otimes\ket{0}_b\) and \(\ket{1}_s\otimes\ket{1}_b\}\). Hence, the unitary action \(U_{sb}\) maps the system-bath joint states \(\{\ket{0}_s\otimes\ket{1}_b\) and \(\ket{1}_s\otimes\ket{0}_b\}\) to any arbitrary pair of orthogonal pure states in the same subspace. Precisely,
\begin{align*}
    U_{sb}\ket{0}_s\otimes\ket{1}_b&=\ket{\psi_{\alpha}^+}_{sb}=(\alpha \ket{01}+\gamma \ket{10})_{sb}\\\text{and } U_{sb}\ket{1}_s\otimes\ket{0}_b&=\ket{\psi_{\alpha}^-}_{sb}=e^{i\phi}(\gamma^* \ket{01}-\alpha^* \ket{10})_{sb},
\end{align*}
where, \(|\alpha|^2+|\gamma|^2=1\) and \(\phi\in[0,2\pi)\). 

Finally, using the bath-reference joint state as in Eq. (\ref{e9}), it is easy to see
\begin{align*}
U_{sb}\otimes\mathbb{I}_R(\ket{0}_s\otimes\ket{\phi_{\beta}^+}_{bR})=&(\frac 1{\sqrt{Z}}(\ket{000}+e^{-\frac{\beta}2}\ket{\psi_{\alpha}^+}\ket{1})_{sbR}\\
U_{sb}\otimes\mathbb{I}_R(\ket{1}_s\otimes\ket{\phi_{\beta}^+}_{bR})=&(\frac 1{\sqrt{Z}}(\ket{\psi_{\alpha}^-}\ket{0}+e^{-\frac{\beta}2}\ket{111}))_{sbR}.
\end{align*}
This identifies the form of isometry in Eq. (\ref{e10}).
\end{proof}
\subsection{Illustrative Examples}\label{exmpl}
In the following, we will discuss few of the illustrative examples with qubit batteries.
\subsubsection{Case. 1}
 Let us consider, \(\alpha=\gamma=\frac1{\sqrt{2}}\) with \(\phi=0\) for \(\ket{\psi^{\pm}_{\alpha}}\) in Eq. (\ref{e10}) for a \(\beta=1\) thermal operation. Then, for an initial QB \(\rho_s=p\ketbra{0}{0}+(1-p)\ketbra{1}{1}\), the system-bath-reference state takes the form
\begin{align}\label{e11}
   \nonumber \sigma_{sbR}&= p\ketbra{\phi_0^1}{\phi_0^1}+(1-p)\ketbra{\phi_1^1}{\phi_1^1}\\\nonumber\text{where, }\ket{\phi_0^1}&=\sqrt{\frac 1{1+e}}(\sqrt{e}\ket{000}_{sbR}+\ket{\psi^+}_{sb}\ket{1}_R)\\\text{and }\ket{\phi_1^1}&=\sqrt{\frac 1{1+e}}(\sqrt{e}\ket{\psi^-}_{sb}\ket{0}_R+\ket{111}_{sbR}).
\end{align}
With further numerical exercise one can show that for the system-reference two qubit state \(\sigma_{sR}=\Tr_b(\sigma_{sbR})\),
\begin{align*}
     E_f(\sigma_{sR})=H(\frac{1+\sqrt{1-C^2(\sigma_{s|R})}}2),\end{align*}
     where, \(H(x)=-x\log x-(1-x)\log (1-x)\) is the Shannon entropy, with base \(e\) and the concurrence
     \begin{align*}
    C(\sigma_{s|R})&=\frac{\sqrt{e}}{2(1+e)}(\sqrt{2+x_p+2\sqrt{2x_p}}\\&-\sqrt{2+x_p-2\sqrt{2x_p}}-2\sqrt{p(1-p)}),
\end{align*}
where, \(x_p=\sqrt{2+p-p^2}\). Accordingly, one can then estimate the bound on the weakly retrieved charges in the qubit battery. The strictly positive nature of \(C(\sigma_{s|R})\) then concludes that the for every \(\rho_s\) with \(p\in(0,1)\), no POVM on the bath qubit is able to retrieve the lost charges back optimally. However, for the zero energetic and the fully charged QB (\(\ket{0}\) and \(\ket{1}\) respectively), simple projective measurements on the bath qubit results in optimal retrieval of the lost charge (see Fig \ref{f2} (a)). 
\begin{figure}
    \centering
    \includegraphics[width=1.0\linewidth]{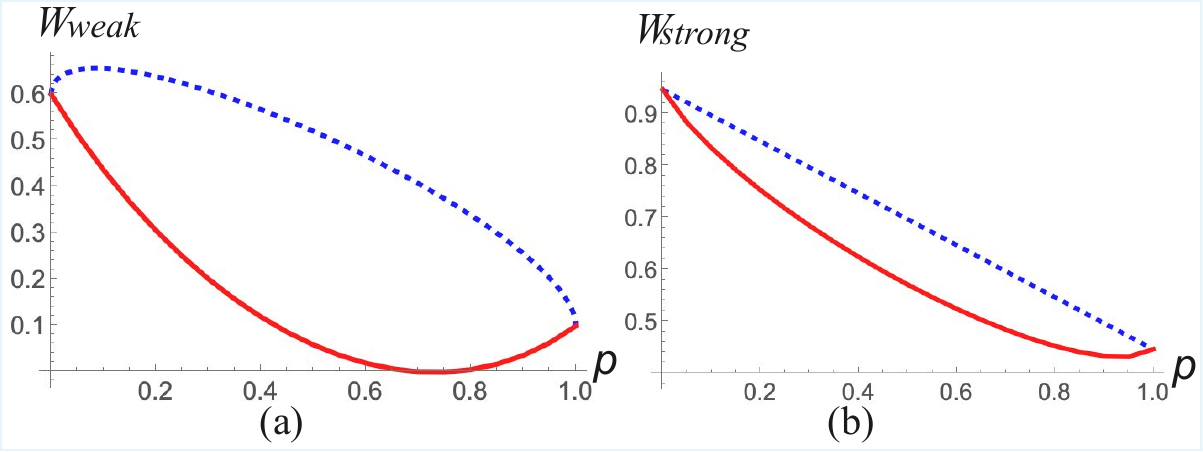}
    \caption{(Color online) \textit{Charge retrieval for QB diagonal in energy Eigen basis.} For the both cases, we have considered the initial QB as \(\rho_s=p\ketbra{0}{0}+(1-p)\ketbra{1}{1},~\beta=1\) and \(\phi=0,~\alpha=\gamma=\frac 1{\sqrt{2}}\).(a) Under weak assistance the \textcolor{blue}{blue} dotted curve denotes the upper bound of charge retrieval. The \textcolor{red}{red} solid curve denotes the amount of truly retrieved charge, optimized over all possible projective measurement performed on the bath qubit. (b) The same quantities are plotted considering the local projective measurements performed on the bath and the reference qubits.}
    \label{f2}
\end{figure}
Notably, both for \(W_{weak}\) and \(W_{strong}\) the origin of the difference between the optimal bounds and the respective quantities by allowing only projective measurements are two folds: While it depicts the necessity of many outcome extreme POVMs to increase the amount of retrieved charges, the optimal bounds are not always achievable. In fact, as depicted in the proof of Theorem \ref{t1} (Appendix \ref{pt1}) the exact amount of retrieved charges under the weak assistance can be quantified as,
\begin{align*}
    W_{weak}(\rho_s,\Lambda_{\beta})= E(\sigma_s)-\frac1{\beta}E_f(\sigma_{s|RR'}).
\end{align*}
Here, \(R'\) is an additional purifying system for the system-bath-reference state \(\sigma_{sbR}\), which is mixed in general. While, \(E_f(\sigma_{s|RR'})\) is general hard to compute, using the monogamy of squared concurrence \cite{coffman2000distributed}, we can write
\begin{align*}
    C(\sigma_{s|RR'})\geq \sqrt{C^2(\sigma_{s|R})+C^2(\sigma_{s|R'})}=:C^*(\sigma_{s|RR'})
\end{align*}
and accordingly we can derive a modified upper-bound as
\begin{align*}
    W_{weak}(\rho_s,\Lambda_{\beta=1})\leq E(\sigma_s)-E_f^*(\sigma_{s|RR'}),
\end{align*}
where \(E_f^*(\sigma_{s|RR'})=H(\frac{1+\sqrt{1-[C^*(\sigma_{s|RR'})]^2}}2)\), as mentioned earlier.

Similarly, for charge retrieval under the strong assistance the upper bound can be modified further by invoking the additional purifying reference \(R'\) (see the proof of Theorem \ref{t2} in Appendix \ref{pt2}), i.e., 
\begin{align*}
    W_{strong}(\rho_s,\Lambda_{\beta})\leq E(\sigma_s)-\frac 1{\beta}E_f(\sigma_{s|R'})\leq E(\sigma_s).
\end{align*}
However, limiting the assistance to be minimal, i.e., local POVMs on the bath and the reference system \(R\), forbids to prepare all possible decompositions of \(\sigma_{sR'}\) and hence to achieves the optimal bound. Interestingly, whenever the initial QB is a pure state the optimal value can be achieved just by performing the local projective measurements on the bath and the reference qubits (\(b\) and \(R\) respectively). 
\subsubsection{Case. 2}
Let us now consider another scenario with the same \(\beta=1, ~\alpha=\gamma=\frac1{\sqrt{2}}\), however with \(\phi=\pi\). Then for a qubit battery diagonal in the energy Eigen basis \(\rho_s=p\ketbra{0}{0}+(1-p)\ketbra{1}{1}\), the system-bath-reference qubit joint state takes the form
\begin{align}\label{a17}
    \nonumber\sigma_{sbR}&= p\ketbra{\phi_0^1}{\phi_0^1}+(1-p)\ketbra{\phi_1^1}{\phi_1^1}\\\nonumber\text{where, }\ket{\phi_0^1}&=\sqrt{\frac 1{1+e}}(\sqrt{e}\ket{000}_{sbR}+\ket{\psi^+}_{sb}\ket{1}_R)\\\text{and }\ket{\phi_1^1}&=\sqrt{\frac 1{1+e}}(-\sqrt{e}\ket{\psi^-}_{sb}\ket{0}_R+\ket{111}_{sbR}).
\end{align}
Using the same analysis as in \textit{Case. 1} above, we can estimate the amount of weakly and strongly assisted retrieved charge. Moreover, as discussed there, involving an additional system \(R'\),to purify the state \(\sigma_{sbR}\), significantly lowers the optimal bounds as proposed in Theorem \ref{t1} and \ref{t2}. The same is depicted in Fig. \ref{f2}.
\begin{figure}[htb]
    \centering
    \includegraphics[width=1.0\linewidth]{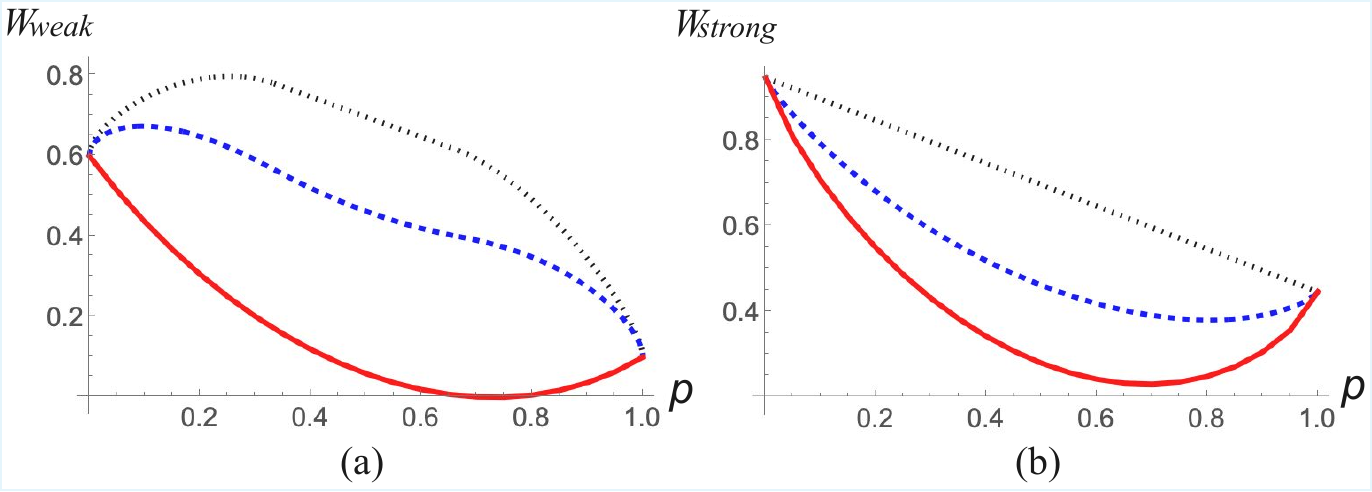}
    \caption{(Color online) \textit{Charge retrieval of QB diagonal in energy  eigenbasis.} The dotted black line denotes the obtained  upper bound from Theorem 1 and Theorem 2 respectively for the plots (a) and (b). The \textcolor{blue}{blue} dashed lines denote further modification on the optimal achievable values, involving an additional purifying system \(R'\). Finally, for (a), the \textcolor{red}{red} curve denotes the amount of retrieved charge by performing the projective measurement only on the bath qubit. For (b), the same \textcolor{red}{red} curve denotes the amount of retrieved charge on the QB by performing local projective measurements on the bath (\(b\)) and the reference qubit (\(R\)).}
    \label{f2}
\end{figure}
\end{document}